%% file: main.tex
\documentclass[runningheads]{llncs}
\input{commands}

\title{Mining Beyond the Bools: Learning Data Transformations and Temporal Specifications}
\titlerunning{Learning   Data Transformations and Temporal Specifications} 

\author{ Sam Nicholas Kouteili\inst{1} \and William Fishell\inst{2} \and 
 Christian Scaff\inst{2} \and 
 Mark Santolucito\inst{2}$^\dagger$ \and Ruzica Piskac\inst{1}}
 \authorrunning{Kouteili et al.}

 \authorrunning{Kouteili et al.}

 \institute{
 Yale University \and
 Columbia University ($\dagger$ Barnard College, Columbia University)
 }

\begin{document}

\maketitle
\vspace{-0.7cm}
\begin{abstract}
    Mining specifications from execution traces presents an automated way of capturing characteristic system behaviors. However, existing approaches are largely restricted to Boolean abstractions of events, limiting their ability to express data-aware properties. In this paper, we extend mining procedures to operate over richer datatypes. 
    We first establish candidate functions in our domain that cover the set of traces
    by leveraging Syntax-Guided Synthesis (SyGuS) techniques. To capture these function applications temporally, we formalize the semantics of TSL$_f$, a finite-prefix interpretation of Temporal Stream Logic (TSL) that extends LTL$_f$ with support for first-order predicates and functional updates. This allows us to unify a corresponding procedure for learning the data transformations and temporal specifications of a system. We demonstrate our approach synthesizing reactive programs from mined specifications on the OpenAI-Gymnasium ToyText environments, finding that our method is more robust and orders of magnitude more sample-efficient than passive learning baselines on generalized problem instances. 
    
\end{abstract}



\input{sections/1-introduction}

\input{sections/2-motivating}
\input{sections/3-formalisms}
\input{sections/4-mining}
\input{sections/6-evaluation}
\input{sections/2-related}

\input{sections/7-conclusion}
\bibliographystyle{splncs04}
\bibliography{citations}

\appendix
\include{appendix/proofs}

\include{appendix/evalextra}

\end{document}

%% file: commands.tex
\usepackage{graphicx} 
\usepackage{soul}
\usepackage{listings}
\usepackage[numbers]{natbib}
\usepackage{amsmath}
\usepackage{hyperref}
\usepackage{url}
\usepackage{cleveref}
\usepackage{subcaption}
\usepackage{algorithmic}
 \usepackage[ruled,vlined]{algorithm2e}
\usepackage{amssymb}
\usepackage{stmaryrd}
\usepackage{booktabs}
\usepackage{enumitem}
\usepackage{float}
\usepackage{tabularx}
\usepackage[table]{xcolor}
\usepackage{collcell}
\usepackage{multirow} 
\usepackage{pgfmath}

\definecolor{strongred}{RGB}{255, 102, 102}
\definecolor{medred}{RGB}{255, 153, 153}
\definecolor{weakred}{RGB}{255, 204, 204}
\definecolor{weakgreen}{RGB}{204, 255, 204}
\definecolor{medgreen}{RGB}{153, 255, 153}
\definecolor{stronggreen}{RGB}{0, 220, 0}

\newcommand{\colorcell}[1]{%
  \def\temp{#1}%
  \def\dash{-}%
  \ifx\temp\dash
    -
  \else
    \ifnum#1<10
      \cellcolor{strongred}#1%
    \else\ifnum#1<20
      \cellcolor{medred}#1%
    \else\ifnum#1<30
      \cellcolor{weakred}#1%
    \else\ifnum#1<40
      \cellcolor{weakgreen}#1%
    \else\ifnum#1<50
      \cellcolor{medgreen}#1%
    \else
      \cellcolor{stronggreen}\textbf{#1}%
    \fi\fi\fi\fi\fi
  \fi
}

\newcommand{\ccp}[2]{%
  \pgfmathtruncatemacro{\mainval}{#1/2}%
  \pgfmathtruncatemacro{\parenval}{#2/2}%
  \ifnum\mainval<10
    \cellcolor{strongred}\mainval{} (\parenval)%
  \else\ifnum\mainval<20
    \cellcolor{medred}\mainval{} (\parenval)%
  \else\ifnum\mainval<30
    \cellcolor{weakred}\mainval{} (\parenval)%
  \else\ifnum\mainval<40
    \cellcolor{weakgreen}\mainval{} (\parenval)%
  \else\ifnum\mainval<50
    \cellcolor{medgreen}\mainval{} (\parenval)%
  \else
    \cellcolor{stronggreen}\textbf{\mainval{} (\parenval)}%
  \fi\fi\fi\fi\fi
}

\newcommand{\cc}[1]{\colorcell{#1}}
\newcommand{\ccd}{-}

\newcommand{\update}[2]{[\texttt{#1} \leftarrow \texttt{#2}]}

\definecolor{jsonkey}{RGB}{0,0,180}    
\definecolor{jsonstring}{RGB}{0,128,0} 
\definecolor{jsonnumber}{RGB}{200,0,0} 
\definecolor{jsonenum}{RGB}{255,140,0} 
\definecolor{jsonframe}{gray}{0.7}     
\definecolor{listinggreen}{RGB}{0,140,0}
\definecolor{listingred}{RGB}{140,0,0}

\lstdefinelanguage{jsonabbr}{
    basicstyle=\ttfamily\tiny,  
    showstringspaces=false,
    breaklines=true,
    frame=single,
    rulecolor=\color{jsonframe},
    morekeywords={
      x, y, goalx, goaly, h0x, h0y, hole0y, h1x, h1y, h2x, h2y,
      cliffHeight,cliffXMax,cliffXMin,
      DEST_x, DEST_y, PASS_x, PASS_y, RED_x, RED_y, GREEN_x, GREEN_y, YELLOW_x, YELLOW_y, BLUE_x, BLUE_y,
      count, stood, standThreshold, standVsWeakMin, isWeakDealer
    },
    keywordstyle=\color{jsonkey}\bfseries,
    literate=
      {:}{{{\color{black}{:}}}}{1}
      {,}{{{\color{black}{,}}}}{1}
  }

\definecolor{tslkw}{RGB}{180,70,0}    
\definecolor{tslop}{RGB}{200,0,0}   
\definecolor{tslpred}{RGB}{90,130,0}  
\definecolor{tslfn}{RGB}{0,150,130}   
\definecolor{tslfrm}{gray}{0.85}

\lstdefinelanguage{tsl}{
  basicstyle=\ttfamily\footnotesize,
  showstringspaces=false,
  breaklines=true,
  frame=single,
  rulecolor=\color{tslfrm},
  morekeywords={G,F,X,U,R,W,∨},
  keywordstyle=\color{tslkw}\bfseries,
  morekeywords=[2]{&&,||,!,->,<->},
  keywordstyle=[2]\color{tslop},
  morekeywords=[3]{isClick,viewProduct,addToCart,checkout,isNav,eq,lt,le,gte,lte,
  rightmost,leftmost, inside_board,has_black_below,has_black_above,has_black_left,has_black_right,surrounded_3white},
  keywordstyle=[3]\color{tslpred},
  morekeywords=[4]{moveLeft,moveRight, nearest_black_y_up,nearest_black_y_down,nearest_black_x_right,nearest_black_x_left, discovered_x_updates,discovered_y_updates,mined_tslf_specification
  },
  keywordstyle=[4]\color{tslfn},
  morecomment=[l]{//},
  commentstyle=\color{gray},
  literate={<-}{{$\leftarrow$}}1 {->}{{$\rightarrow$}}1
}

%% file: sections/1-introduction.tex
\section{Introduction}\label{sec:intro}



\emph{Specification mining} is the process of deriving logical properties from demonstrations, discovering invariants that unify observations or discriminate positive and negative examples~\cite{glenn2002mining, lemieux2015general, neider2018learning}. Manually crafting specifications is an error-prone process, with numerous examples where flawed specifications have led to serious software failures. For instance, faulty specifications in the Linux kernel have resulted in exploitable vulnerabilities~\cite{dossche2024usenix}, while specification gaps have caused critical failures in large-scale cloud systems~\cite{newcombe2015amazon}. To reduce reliance on human intuition, specification mining automates learning directly from executions.

Standard approaches to specification mining analyze execution traces, which capture system behavior over time. As a result, temporal logics are used as a natural formalism for expressing the specification.
Linear Temporal Logic (LTL) \cite{pnueli1977temporal,clarke1986automatic} and its finite-prefix interpretation LTL$_f$~\cite{giacomo2013ltlf} have long been the standard for encoding temporal properties of systems, with applications in agent-based planning~\cite{bacchus2000using,aminof2025multiagent}
and runtime verification \cite{clarke1999model}.
Recent mining tools such as Scarlet~\cite{raha2024Scarlet} and Bolt~\cite{bathie2025bolt} demonstrate the feasibility of mining LTL$_f$ specifications.

Despite their adoption, LTL and LTL$_f$ are propositional formalisms, meaning that all events are encoded as Booleans. This results in multiple challenges. For example, to express properties that depend on variable evolutions (i.e. “move $x$ to the cell above the nearest obstacle”), we can either hand-craft predicates and their interpretations (requiring manual intervention) or bit-blast the trace's data (increasing formula size and potentially introducing spurious semantic relations).
Temporal Stream Logic (TSL)~\cite{santolucito2019tsl} presents a more expressive alternative that addresses this. By separating control from data, TSL allows specifications to express properties not only about the temporal ordering of events, but also about the data transformations applied to variables defined over arbitrary types.

In this paper, we demonstrate that a temporal logic equipped with data transformations offers an expressive formalism for capturing complex specifications.
Because specification mining operates over finite traces, we introduce TSL$_f$, a finite-prefix interpretation of TSL.
TSL$_f$ improves on the expressivity of LTL$_f$ by natively capturing functional updates on variables over time. 

Our specification mining algorithm operates on sets of input traces whose variables may range over arbitrary types, not just Booleans. To infer a specification, we must first discover functions that change system variables over time, and 
 the temporal relations among those variables. 
We use Syntax-Guided Synthesis (SyGuS)~\cite{alur2013syntax} for the automatic discovery of functions from input-output examples. 
To identify a set of functions that captures how variables evolve, our procedure synthesizes a function for every variable at every time step.
We next iteratively prune this set until we find functions that minimally describe the data transformations over the entire trace.
Finally, using the learned set of transformations, we lift our data traces to function application traces.
Such traces are then amenable to standard mining algorithms such as Boolean Subset Cover~\cite{bathie2025bolt}.






We evaluate our approach on the OpenAI-Gymnasium ToyText suite~\cite{openai2016gym}, a collection of environments in which agents must reason temporally to reach a goal state. Using only positive and negative traces, and without any domain-specific knowledge, action functions, or fine-grained reward signals, our method extracts specifications that fully capture the environment objectives. 
We validate these mined specifications by synthesizing reactive controllers and deploying them as game-playing agents, achieving perfect win rates on held-out configurations using only a small number ($\leq 20$) of examples. 
Our approach significantly outperforms both symbolic and neural baselines, requiring an order of magnitude fewer examples while demonstrating better performance on generalized problems.

Overall, this work enables the mining of more expressive temporal properties than prior approaches; more broadly, it presents a step toward a purely symbolic reinforcement learning (RL) paradigm, where an agent interacts with an environment, learns from mining generated traces, and refines itself through formal specification.
In summary, our contributions are as follows:
\vspace{-0.1cm}
  \begin{itemize}
      \item We present a bottom-up synthesis algorithm that discovers function sets that cover full execution traces, with a greedy input-swapping strategy.
      \item We introduce a specification mining framework over functions and datatypes using TSL$_f$, a finite-trace semantics for Temporal Stream Logic.
      \item We show on standard RL benchmarks that from mined specifications, we can synthesize reactive programs that are 
      sample-efficient and generalizable.
  \end{itemize}

%% file: sections/2-motivating.tex
\begin{figure}[t]
    \centering
    \begin{subfigure}[h]{0.54\linewidth}
        \centering
        \includegraphics[width=\linewidth]{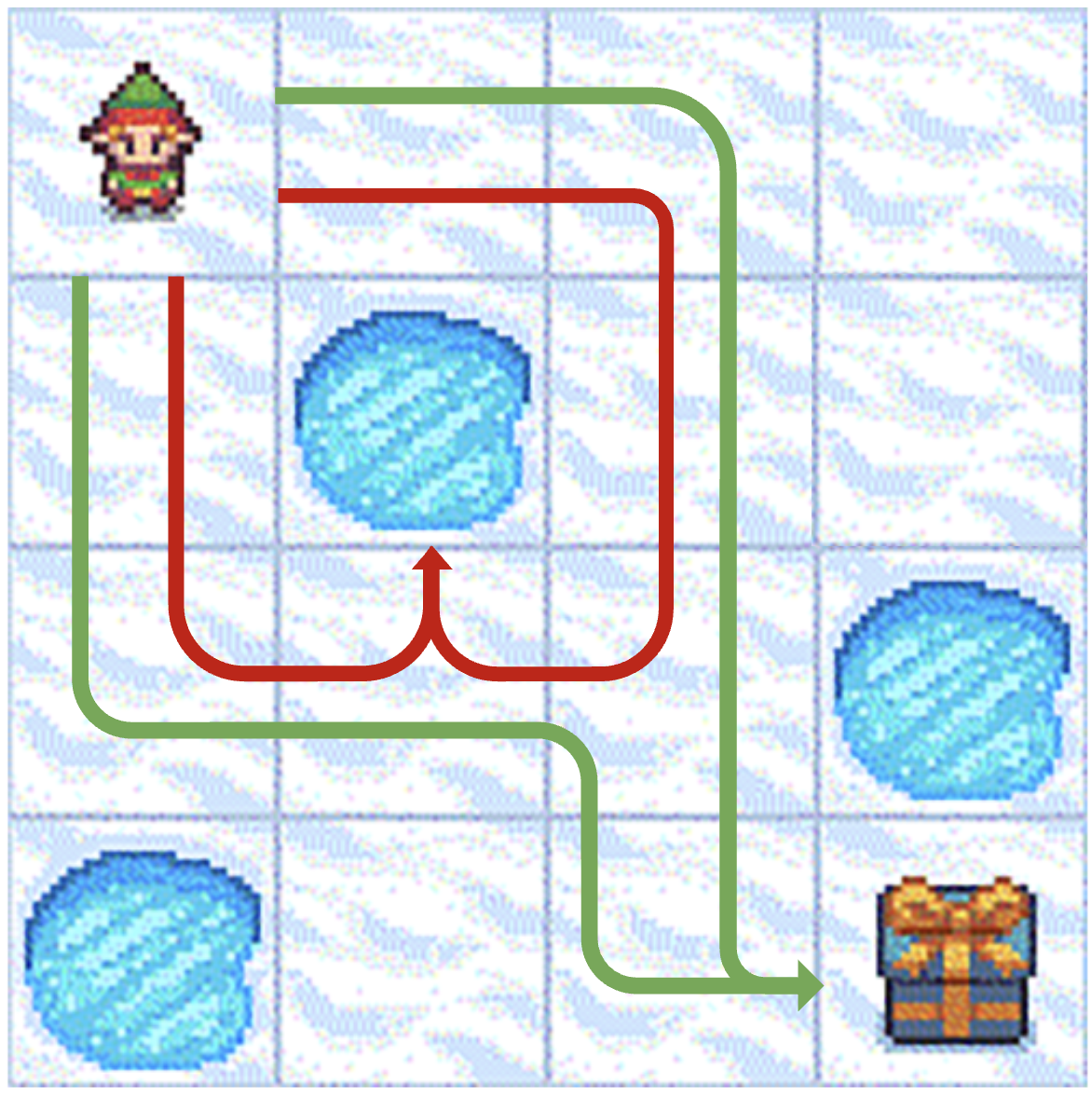}
    \end{subfigure}
    \hfill
    \begin{subfigure}[h]{0.445\linewidth}
        \centering
    \begin{tabular}{c}
         \begin{lstlisting}[
    language=jsonabbr,
    float=false,
    basicstyle=\ttfamily\tiny,
    aboveskip=0pt,
    belowskip=0pt,
    frame=single,
    rulecolor=\color{listinggreen}
]
{p:(0,0), g:(3,3), h0:(1,1), h1:(0,3), h2:(3,2)}
{p:(0,1), g:(3,3), h0:(1,1), h1:(0,3), h2:(3,2)}
{p:(0,2), g:(3,3), h0:(1,1), h1:(0,3), h2:(3,2)}
{p:(1,2), g:(3,3), h0:(1,1), h1:(0,3), h2:(3,2)}
{p:(2,2), g:(3,3), h0:(1,1), h1:(0,3), h2:(3,2)}
{p:(2,3), g:(3,3), h0:(1,1), h1:(0,3), h2:(3,2)}
{p:(3,3), g:(3,3), h0:(1,1), h1:(0,3), h2:(3,2)}
\end{lstlisting} \\

         \begin{lstlisting}[
    language=jsonabbr,
    float=false,
    basicstyle=\ttfamily\tiny,
    aboveskip=0pt,
    belowskip=0pt,
    frame=single,
    rulecolor=\color{listinggreen}
]
{p:(0,0), g:(3,3), h0:(1,1), h1:(0,3), h2:(3,2)}
{p:(1,0), g:(3,3), h0:(1,1), h1:(0,3), h2:(3,2)}
{p:(2,0), g:(3,3), h0:(1,1), h1:(0,3), h2:(3,2)}
{p:(2,1), g:(3,3), h0:(1,1), h1:(0,3), h2:(3,2)}
{p:(2,2), g:(3,3), h0:(1,1), h1:(0,3), h2:(3,2)}
{p:(2,3), g:(3,3), h0:(1,1), h1:(0,3), h2:(3,2)}
{p:(3,3), g:(3,3), h0:(1,1), h1:(0,3), h2:(3,2)}
\end{lstlisting} \\
         \begin{lstlisting}[
    language=jsonabbr,
    float=false,
    basicstyle=\ttfamily\tiny,
    aboveskip=0pt,
    belowskip=0pt,
    frame=single,
    rulecolor=\color{listingred}
]
{p:(0,0), g:(3,3), h0:(1,1), h1:(0,3), h2:(3,2)}
{p:(1,0), g:(3,3), h0:(1,1), h1:(0,3), h2:(3,2)}
{p:(2,0), g:(3,3), h0:(1,1), h1:(0,3), h2:(3,2)}
{p:(2,1), g:(3,3), h0:(1,1), h1:(0,3), h2:(3,2)}
{p:(2,2), g:(3,3), h0:(1,1), h1:(0,3), h2:(3,2)}
{p:(1,2), g:(3,3), h0:(1,1), h1:(0,3), h2:(3,2)}
{p:(1,1), g:(3,3), h0:(1,1), h1:(0,3), h2:(3,2)}
\end{lstlisting}
 \\
         \begin{lstlisting}[
    language=jsonabbr,
    float=false,
    basicstyle=\ttfamily\tiny,
    aboveskip=0pt,
    belowskip=0pt,
    frame=single,
    rulecolor=\color{listingred}
]
{p:(0,0), g:(3,3), h0:(1,1), h1:(0,3), h2:(3,2)}
{p:(0,1), g:(3,3), h0:(1,1), h1:(0,3), h2:(3,2)}
{p:(0,2), g:(3,3), h0:(1,1), h1:(0,3), h2:(3,2)}
{p:(1,2), g:(3,3), h0:(1,1), h1:(0,3), h2:(3,2)}
{p:(1,1), g:(3,3), h0:(1,1), h1:(0,3), h2:(3,2)}
\end{lstlisting}
    \end{tabular}
        
    \end{subfigure}

    \caption{\textsc{FrozenLake}~\cite{openai2016gym}. Two positive (green) and two negative (red) traces.}
    \label{fig:frozenlake-traces}
    \vspace{-0.3cm}
\end{figure}

\section{Motivating Example}\label{sec:motiv}



Modern assessments of AI systems emphasize out-of-distribution generalization, i.e. the ability to apply learned concepts to novel situations not encountered during training. The ARC-AGI benchmarks have become a prominent testbed for this capability, presenting tasks on discrete grids where success requires inferring latent rules from sparse examples \cite{chollet2019arc}. While humans achieve 80–90\% accuracy on these tasks, state-of-the-art large language models plateau far below, only reaching 31\% \cite{arcagi2025report}. The recently introduced ARC-AGI3 extends this challenge to ``interactive reasoning benchmarks"~\cite{kamradt2025arcagi3}: grid-based games 
with actions modifying state over time and success dependent on inferring reactive rules. Such benchmarks fall into a class of temporal and relational reasoning problems. 

We illustrate our approach on \textsc{FrozenLake}, a canonical grid-world task from the OpenAI-Gymnasium ToyText suite~\cite{openai2016gym} also in this temporal and relational class. As shown in \cref{fig:frozenlake-traces}, an agent begins at a start cell and must reach a goal while avoiding holes on the surface. 
Suppose we have two positive traces (successfully reaching the goal) and two negative traces (ending in a hole). Each trace consists of tuples recording the 
values of five variables {\texttt{(p, g, h0, h1, h2)}}. Here, \texttt{p} refers to the player, \texttt{g} to the goal, and \texttt{h*} to the holes. We are only given the traces, with no domain-specific functions or predicates.  

Our pipeline first uncovers the functions that exist in the game: the player's coordinates either increment \texttt{(+1)} or decrement \texttt{(-1)}. Our mining algorithm then derives both a liveness specification, capturing what must eventually happen, and a safety specification, stating what must always hold. 
We recover the liveness formula,
which expresses that the agent must eventually reach the goal. However, from these four traces, the mined safety specification states that the agent may not move upwards.
(Note: {\texttt{p[2]}} is the second tuple value of {\texttt{p}}).
$$
  \varphi_{\text{live}} = \mathbf{F} (\texttt{(=)}\ \texttt{p}\ \texttt{g}), \ \ 
  \varphi_{\text{safe}}^{(1)} = \mathbf{G} \neg  (\update{p[2]}{(-1) p[2]})
$$
This safety specification is clearly  spurious, yet it still distinguishes positive from negative traces in the given sample. If we synthesize a reactive controller from $\varphi_{\text{live}} \land \varphi_{\text{safe}}^{(1)}$ and deploy it, the agent indeed proceeds straight downwards and falls into a hole. 
This fall generates a new negative trace. Adding it to our corpus and re-mining, we obtain the same liveness formula but a refined safety specification:
$$
\varphi_{\text{live}} = \mathbf{F} (\texttt{(=)}\ \texttt{p}\ \texttt{g}),  \ \varphi_{\text{safe}}^{(2)}  = \mathbf{G} \neg (\texttt{(=)}\ \texttt{p}\ \texttt{h0} \vee \texttt{(=)}\ \texttt{p}\ \texttt{h1})
$$
This specification captures the nature of the constraint (avoid holes) but only identifies failures at \texttt{hole0} and \texttt{hole1}. Synthesizing a controller from this specification results in an agent avoiding these two holes but falling into \texttt{hole2}, which was never visited. Adding this negative sample, we extract a correct specification
\begin{align*}
   \varphi_{\text{live}}  = \mathbf{F} (\texttt{(=)}\ \texttt{p}\ \texttt{g}), \ \  
  \varphi_{\text{safe}}^{(3)}  = \mathbf{G} \neg (\texttt{(=)}\ \texttt{p}\ \texttt{h0} \vee \texttt{(=)}\ \texttt{p}\ \texttt{h1} \ \vee \ \texttt{(=)}\ \texttt{p}\ \texttt{h2})
\end{align*}

The synthesized controller now wins consistently. Importantly, because the specification is relationally expressed (comparing player coordinates to hole coordinates rather than memorizing concrete positions), it generalizes to boards with different hole placements and different grid sizes never seen during mining. 

In essence, this demonstrates a form of \emph{symbolic reinforcement learning}. 
Rather than fitting a policy through gradient descent, we synthesize a formal world model through specification mining over traces. The result is an agent whose behavior is constrained by discovered temporal and relational invariants. 

%% file: sections/3-formalisms.tex
\vspace{-0.4cm}
\section{Formalisms}\label{sec:formalism}

We introduce TSL$_f$, the finite variant of TSL.
In order to learn temporal specifications over data transformations, we require a logic that natively supports temporal operators and reasoning over theories. There has been an increasing interest in such logics recently, including for example, LTL$_f$-MT~\cite{rodriguez2023boolean,rodriguez2024adaptive,rodriguez2025counter} or Reactive
Program LTL~\cite{heim-popl25, heim2025issy}.
We also focus on TSL$_f$ in particular because it is well supported from a tool perspective for the construction of reactive programs. 

\vspace{-0.4cm}
\subsection{Syntax}
Under a finite-prefix interpretation of behaviors, the operators of TSL$_f$ do not vary from their infinite counterparts; however, their semantics do. 

TSL$_f$ is parameterized by a signature $\Sigma = (S, F, P, \Delta)$ consisting of:
\begin{itemize}
  \item A finite set $S$ of typed \emph{stream variables}, which represent data evolution.
  \item A finite set $F$ of \emph{function symbols}, each with a type signature $\tau_1 \times \cdots \times \tau_n \to \tau$.
  \item A finite set $P$ of \emph{predicate symbols}, with a type signature $\tau_1 \times \cdots \times \tau_n \to \mathtt{Bool}$.
  \item A background theory $\Delta$ over which functions and predicates are interpreted.
\end{itemize}

We follow the TSL-MT framework from \citep{santolucito2022tslmt} for interpretation and application of theories. Our experiments consider $\Delta$ to be the theory of Linear Integer Arithmetic (LIA), but other theories, such as the
the theory of arrays, are also applicable. Function and predicate symbols are interpreted as concrete implementations over this theory, which we discover with our mining procedure.

 In TSL$_f$, we introduce a distinguished atomic proposition $\mathtt{END} \in \mathcal{P}$ that is true exclusively at the final position of a trace. This is dual to the \texttt{Alive} predicate used in translations from LTL$_f$ to standard LTL over infinite traces~\cite{giacomo2013ltlf}. 

\vspace{0.5em}
\noindent\textbf{Terms.} We define two classes of terms over signature $\Sigma$:
\begin{itemize}
  \item \emph{Update terms} ($\mathcal{T}_u$) $[s \gets f(s_1, \ldots, s_n)]$ where $s, s_1, \ldots, s_n \in S$ and $f \in F$ is a function whose type signature matches the types of $s_1, \ldots, s_n$ and returns a value of the same type as $s$. $[s \gets f(s_1, \ldots, s_n)]$ states that $s$ is updated in the \emph{next} time step by applying function $f$ to the current values of $s_1, \ldots, s_n$.

\item \emph{Predicate terms} ($\mathcal{T}_p$) $p(s_1, \ldots, s_n)$ where $p \in P$ and $s_1, \ldots, s_n \in S$ are streams whose types match the signature of $p$. It states that a predicate is true over its terms at the \emph{current} timestep. Predicate terms include $\mathtt{END}$.
\end{itemize}

\noindent\textbf{Formulas.} TSL$_f$ formulas are defined by the following grammar:
\begin{equation}
  \varphi ::= \tau \in \mathcal{T}_u \cup \mathcal{T}_p \ | \ \neg \varphi \ | \ \varphi \land \varphi \ | \ \mathbf{X}  \varphi \ | \ \varphi  \mathbf{U}  \varphi
  \label{eq:tsl}
\end{equation}
 where $\mathbf{X}$ (neXt) and $\mathbf{U}$ (Until) are the standard LTL$_f$ operators.
We note standard abbreviations commonly used:
\begin{itemize}
  \item Boolean connectives: $\varphi_1 \vee \varphi_2 \equiv \neg(\neg\varphi_1 \land \neg\varphi_2)$ and $\varphi_1 \to \varphi_2 \equiv \neg\varphi_1 \vee \varphi_2$
  \item Temporal operators: $\mathbf{F}  \varphi \equiv \top  \mathbf{U}  \varphi$ (eventually) and $\mathbf{G}  \varphi \equiv \neg \mathbf{F} \neg \varphi$ (always)
\end{itemize}

\subsection{Semantics}

We interpret TSL$_f$ formulas over finite traces. A \emph{trace}  a finite sequence $\sigma = \sigma_0 \sigma_1 \cdots \sigma_n$ where each $\sigma_i \subseteq \mathcal{T}_u \cup \mathcal{T}_p$ is a set of terms that hold at position $i$. 

\noindent\textbf{Well-Formed Trace.} A trace $\sigma$ is \emph{well-formed} with respect to updates in TSL$_f$ if it satisfies the following constraints:
\begin{enumerate}
  \item For each position $i < |\sigma| - 1$ (i.e., every non-final position), every stream variable is uniquely updated. This follows from the mutual exclusion of updates, which intuitively states that a variable cannot be set to two different terms at the same timestep~\citep{santolucito2019frp}:
  $$\forall i \in [0,|\sigma|-1) \ \forall s \in S \   \exists! \   u_s \in \mathcal{T}_u. \   u_s \in \sigma_i$$
\item At the final position $|\sigma| - 1$, no updates hold. Intuitively, there is no clear determination of a stream variable to a next timestep that is undefined:
$$\forall  u \in \mathcal{T}_u. \   u \notin \sigma_{|\sigma|-1}$$
\item The distinguished predicate $\mathtt{END}$ holds exclusively at the final position. 
$$\mathtt{END} \in \sigma_i \iff i = |\sigma| - 1$$
\end{enumerate}


\noindent\textbf{Satisfaction.} The satisfaction relation $\sigma, i \models \ \varphi$ (``formula $\varphi$ holds at position $i$ of trace $\sigma$'') is defined recursively. Trace $\sigma$ \emph{satisfies} formula $\varphi$ if $\sigma, 0 \models \ \varphi$:
  \begin{align*}
      \sigma, i \models \ & [s \gets f(s_1, \ldots, s_n)] \iff [s \gets f(s_1, \ldots, s_n)] \in \sigma_i \\
      \sigma, i \models \ & p(s_1, \ldots, s_n) \iff p(s_1, \ldots, s_n) \in \sigma_i \\ 
      \sigma, i \models \ & \mathtt{END} \iff i = |\sigma| - 1 \\
      \sigma, i \models \ & \neg \varphi \iff \sigma, i \not\models \ \varphi \\ 
      \sigma, i \models \ & \varphi_1 \land \varphi_2 \iff \sigma, i \models \ \varphi_1 \text{ and } \sigma, i \models \ \varphi_2 \\
      \sigma, i \models \ & \mathbf{X}  \varphi \iff i + 1 < |\sigma| \text{ and } \sigma, i+1 \models \ \varphi \\
      \sigma, i \models \ & \varphi_1 \, \mathbf{U} \, \varphi_2 \iff \exists j \geq i. \, \big(j < |\sigma| \, \land \, \sigma, j \models \varphi_2 \, \land \, \forall k. \, (i \leq k < j \to \sigma, k \models \varphi_1)\big)
  \end{align*}

\subsection{Realizability}

In the context of reactive synthesis, a specification $\varphi$ is \emph{realizable} if there exist functions computing output updates from input/output history such that $\varphi$ is satisfied for all possible input behaviors. For TSL, this manifests by partitioning streams into \emph{input streams} $I \subseteq S$ (controlled by the environment) and \emph{output streams} $O = S \setminus I$ (controlled by the system) \citep{santolucito2019tsl}.
For TSL$_f$, the realizability problem is similarly defined, but over finite-horizon strategies. The system must synthesize update functions that guarantee satisfaction of $\varphi$ for any finite input.

The realizability problem of TSL$_f$, similarly to its infinite counterpart, is undecidable. In \cite{santolucito2019tsl}, the Post Correspondence Problem (PCP) is reduced to a TSL specification, invariably constructing an undecidable realizability instance. This proof holds under TSL$_f$ semantics, as TSL$_f$ is defined under a finite but arbitrary prefix interpretation. A strictly bounded formulation, with an explicit maximal trace length, is decidable, but this is outside the scope of this paper. 


While TSL$_f$ inherits the grammar TSL, the finite-trace interpretation imposes fundamental limitations on realizability. Specifically, TSL$_f$ is strictly less expressive than TSL with respect to the class of realizable specifications.

\begin{theorem}
    TSL$_f$ is strictly less expressive than TSL
\end{theorem}

This gap is also present for LTL$_f$ and LTL~\cite{giacomo2013ltlf}.
The distinction arises from the interaction between the semantics of temporal operators over finite traces and the update mutual exclusion constraint. Consider the following TSL specification:
$$
\mathbf{G} \ (\ \mathbf{F} \ [x \leftarrow y] \ \land \ \mathbf{F} \  [x \leftarrow z] \ ).
$$

Under standard TSL semantics, this is realizable as a controller that alternates between $[x \leftarrow y]$ and $[x \leftarrow z]$ indefinitely. However, under a finite-prefix interpretation, $\mathbf{G} ( \mathbf{F} [x \leftarrow y])$ necessitates that at the last timestep, $[x \leftarrow y]$ is true, as from every step in the finite trace there must eventually be a later step where the update holds. Similarly, $\mathbf{G} ( \mathbf{F} [x \leftarrow z])$ requires $[x \leftarrow z]$ be true at the last timestep. This presents a violation of well-formed semantics, specifically the mutual exclusion of updates, which dictates that at a given timestep, a stream variable can only be updated once. The formal proof is presented in \Cref{sec:thm1}. 





%% file: sections/4-mining.tex
\begin{figure}[t]
    \centering
    \begin{lstlisting}[
    language=jsonabbr,
    float=false,
    basicstyle=\ttfamily\scriptsize,
    aboveskip=0pt,
    belowskip=0pt
]
{x: 0, y: 0, gX: 3, gY: 3, h0X: 1, h0Y: 1, h1X: 3, h1Y: 1, h2X: 3, h2Y: 2} 
{x: 0, y: 1, gX: 3, gY: 3, h0X: 1, h0Y: 1, h1X: 3, h1Y: 1, h2X: 3, h2Y: 2} 
{x: 0, y: 2, gX: 3, gY: 3, h0X: 1, h0Y: 1, h1X: 3, h1Y: 1, h2X: 3, h2Y: 2} 
{x: 1, y: 2, gX: 3, gY: 3, h0X: 1, h0Y: 1, h1X: 3, h1Y: 1, h2X: 3, h2Y: 2} 
{x: 2, y: 2, gX: 3, gY: 3, h0X: 1, h0Y: 1, h1X: 3, h1Y: 1, h2X: 3, h2Y: 2} 
{x: 2, y: 3, gX: 3, gY: 3, h0X: 1, h0Y: 1, h1X: 3, h1Y: 1, h2X: 3, h2Y: 2} 
{x: 3, y: 3, gX: 3, gY: 3, h0X: 1, h0Y: 1, h1X: 3, h1Y: 1, h2X: 3, h2Y: 2}
\end{lstlisting}
    \caption{OpenAI-Gymnasium \textsc{FrozenLake} game trace. Each line represents a timestep and encodes environment state. \texttt{x} and \texttt{y} refer to player coordinates, \texttt{gX} and \texttt{gY} to the goal coordinates, and \texttt{h*X} and \texttt{h*Y} to obstacle coordinates.
    }
    \label{fig:frozenlake}
    \vspace{-0.4cm}
\end{figure}

\section{Mining Procedure}\label{sec:mining}
Specification mining is the task of inferring a specification $\phi$ from demonstrations, such that positive traces satisfy $\phi$ and negative traces violate it. LTL$_f$ specification mining operates on sets of Boolean traces of the form 
$\alpha_1\land \alpha_2;\ \alpha_2 \land \alpha_3;\ \alpha_1\dots,
$
where $\alpha_i$ denotes an atomic proposition and semicolons separate timesteps.
This is limited in practice because real systems involve complex functions, predicates, and temporal interactions defined over non-Boolean data types.
Consider, for example, a log of execution from the OpenAI-Gymnasium ToyText \textsc{FrozenLake} Game (\cref{fig:frozenlake}).
Here, before even attempting to mine system temporal behaviors, we need to establish how data values update over time. Specifically, we first need to posit what functions with which inputs have been applied to different variables, before then being able to extract temporal dependencies of these behaviors.

\subsection{Function Discovery}

Given a set of traces, the function discovery problem finds a set of functions $F$ that explains variable evolutions over the entire set. 
This is a prerequisite for data-rich specification mining: without knowing which functions explain the data transformations, we cannot construct update terms defining our variables.
For every time step but the last one, every variable $v\in S$ in our environment has an update term $[v \leftarrow f  \vec{w}], \ f \in F, \vec{w} \in S^{|\vec{w}|}$ at every non-terminal timestep. 

\begin{figure}[h]
    \centering
    \includegraphics[width=\linewidth]{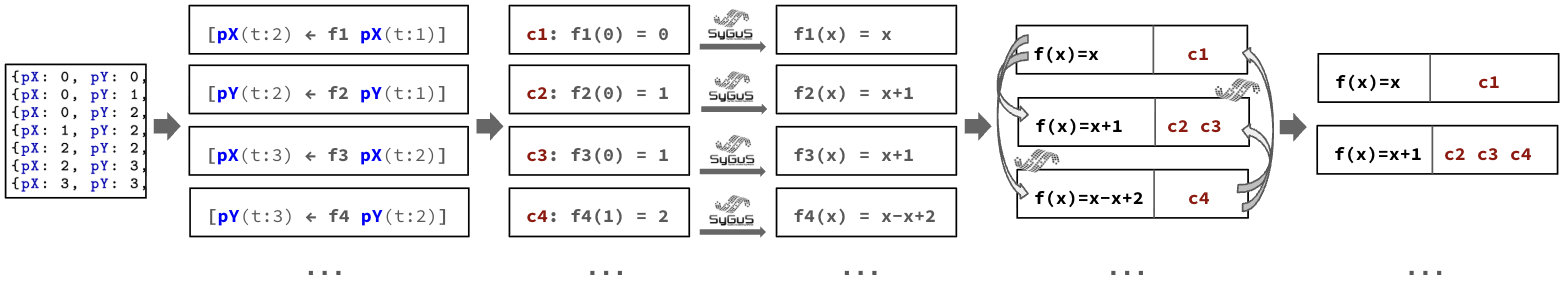}
    \caption{Function discovery procedure. After choosing input-output pairs, the SyGuS solver determines unique functions for each constraint. Constraints are merged by their synthesized functions, at which point smaller constraint sets attempt to merge into larger ones, eliminating spurious functions.}
    \label{fig:func-discovery}
    \vspace{-0.3cm}
\end{figure}

We present a \emph{greedy bottom-up} function discovery procedure. 
We begin with the singleton partition of our traces, treating each variable transition as a constraint with its own function. We then iteratively merge groups that can be unified under one function, exploring alternative function inputs if a merge is unsuccessful. The algorithm terminates when no further merges succeed, yielding a complete (with respect to the theory grammar) and locally minimal function set that covers all variables at all time steps. This greedy strategy avoids exhaustive enumeration while still discovering compact function sets.

\begin{algorithm}[t]
  \caption{SynthesizeCoveringFunctions}
  \label{alg:main-synthesis}
  \KwIn{Logs $\Pi$, Variables $S$, $k_{\max}$ (max arity)}
  \KwOut{Function set $F$, success flag $s$}
  \ForEach{$\pi \in \Pi$}{
    \tcp{Phase 1: Extract Transitions for $\pi$}
    $T_\pi \gets S \times [1,|\pi|)$\;
    \ForEach{$(v,t) \in T_\pi$}{
        $\texttt{alts}_\pi[(v,t)] \gets \{( \texttt{val}^\pi_{t-1}(\vec{w}), \texttt{val}^\pi_{t}(v)\}) \mid \vec{w} \in S^{\ell}, \, \ell \in [1, k_{\max}]\}$\;
        $\texttt{rec}_\pi[(v,t)] \gets \texttt{alts}_\pi[(v,t)][0]$\;
    }
    \tcp{Phase 2: Singleton Synthesis}
    \ForEach{$(v,t) \in T_\pi$}{
        $f_\pi[(v,t)] \gets \texttt{SyGuS}(\texttt{rec}_\pi[(v,t)])$\;
    }
    \tcp{Phase 3: Group and Merge}
    $G_\pi = \{\{(v,t) \in T_\pi \mid f_\pi[(v,t)] = f \} \mid f \in \text{Image}(f_\pi)\}$\;
    $F_\pi \gets \texttt{BottomUpMerge}(G_\pi, \texttt{rec}_\pi, \texttt{alt}_\pi, m_{\max})$\; 
  }
  \Return{$\bigcup_{\pi \in \Pi} F_\pi$}
  \end{algorithm}

The core primitive underlying our algorithm is \emph{Syntax-Guided Synthesis} (SyGuS)~\citep{alur2013syntax}. A SyGuS solver (we use CVC5~\citep{barbosa2022cvc5}) takes as input a set of constraints $\{f(x_1) = y_1, \ldots, f(x_k) = y_k\}$ and a grammar defining the space of candidate functions, and constructs a function $f$ satisfying all constraints. 
For instance, given examples $\{(0,1), (1,2), (2,3)\}$, CVC5 synthesizes $f(x) = x + 1$. 
If the SyGuS solver fails to find a function satisfying the input-output constraints, we suppose no such function exists, and we must choose a different set of constraints that cover the trace. Empirically, we find 0.1 seconds to be a sufficient timeout. 

Algorithm~\ref{alg:main-synthesis} presents our bottom-up synthesis procedure. The algorithm takes as input a set of traces, a maximum function arity, and a maximum number of functions, and outputs a set of discovered functions. The algorithm proceeds in three phases: (1) initializing candidate variable transitions and generating input groupings, (2) singleton synthesis to find independent functions explaining each chosen transition, and (3) grouping transition constraint sets.
In the first phase, the algorithm establishes, for every time step, the possible inputs (variables at the previous timestep) for a given variable. 
Choosing a configuration, we independently synthesize unique functions for each transition.

\begin{algorithm}[h]
  \caption{BottomUpMerge}
  \label{alg:merge}
  \KwIn{Groups $G$, records $\texttt{rec}$, alternatives $\texttt{alts}$}
  \KwOut{Function set $F$}
  $G \gets \texttt{sort}(G, \lambda x.y.|x|<|y|)$ \tcp*{sort $G$ by smallest to largest sets}
  \Repeat{$\neg \textit{merged}$}{
    $\textit{merged} \gets \texttt{false}$\;
    \ForEach{$g_i \in G$}{
        \tcp{Try merging with current input assignments}
        \ForEach{$g_j \in G$ where $|g_j| > |g_i|$}{
            $R \gets \{\texttt{rec}[(v,t)] \mid (v,t) \in g_i \cup g_j\}$\;
            \If{$\texttt{SyGuS}(R) \neq \texttt{null}$}{
                $g_j \gets g_j \cup g_i$; $g_i \gets \emptyset$\;
                $\textit{merged} \gets \texttt{true}$; \textbf{break}\;
            }
        }
        \tcp{If merge failed, try alternative input sources}
        \If{$g_i \neq \emptyset$}{
            \ForEach{$(v,t) \in g_i$}{
                \ForEach{$r \in \texttt{alts}[(v,t)]$}{
                    $\texttt{rec}[(v,t)] \gets r$ \tcp*{Swap to alternative}
                    \ForEach{$g_j \in G$ where $|g_j| > |g_i|$}{
                        $R \gets \{\texttt{rec}[(v,t)] \mid (v,t) \in g_i \cup g_j\}$\;
                        \If{$\texttt{SyGuS}(R) \neq \texttt{null}$}{
                            $g_j \gets g_j \cup g_i$; $g_i \gets \emptyset$\;
                            $\textit{merged} \gets \texttt{true}$; \textbf{break}; \tcp*{break to repeat}
                        }
                    }
                }
            }
        }
    }
  }
  \Return{$\{f_i \mid g_i \in G, g_i \neq \emptyset\}$\;}
  \end{algorithm}

Algorithm~\ref{alg:merge} implements the merging logic. Starting from the singleton partition, the algorithm greedily attempts to unify smaller groups with larger ones by checking whether a single function can explain both sets of constraints. For example, for a given constraint \texttt{f(1)=2}, CVC5 sometimes discovers a spurious function such as $f(x)=(x-x+2)$. Combining this input-output pair with other \texttt{(+1)} pairs weeds out the spurious case. 
We prioritize merging smaller groups first as these are more often spurious relations. When a direct merge fails with the current input source assignments, the algorithm explores alternative input variables for transitions in the smaller group, stipulating that the variable was updated as a result of some other input. This corresponds to trying different hypotheses about which variables are the inputs to the output variable at the time step. 
This process repeats until no further merges are found for the current set of group configurations, yielding a locally minimal function set.

The algorithm is guaranteed to find a function set explaining all observed transitions. In the worst case, if no merges succeed, the algorithm returns one function per transition (the singleton partition). The greedy merging strategy does not guarantee global minimality; it is possible that a different input source assignment or merge order could yield fewer functions, but the bottom-up merging procedure of Algorithm~\ref{alg:merge} tractably tends toward compact function sets.


  \begin{figure}[t]
    \centering
    \includegraphics[width=0.85\linewidth]{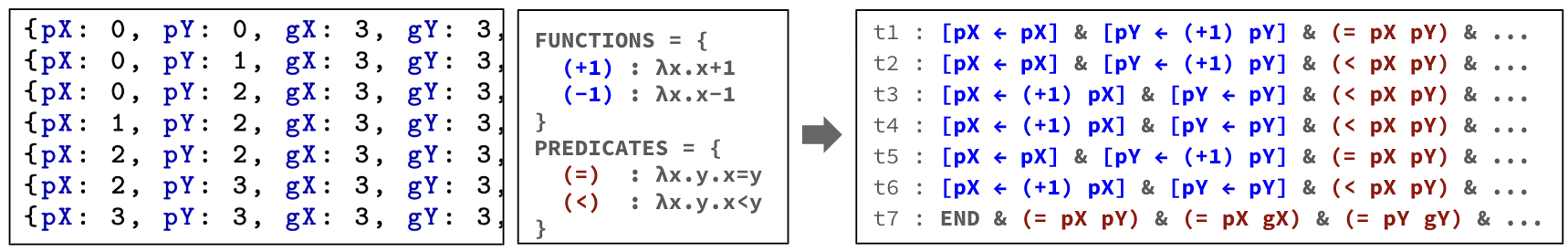}
    \caption{Constructing well-formed TSL$_f$ traces. Having learned the functions of our system, we model a valid interpretation of system behavior. For the trace to be well-formed,  each variable must only have one update (until the last timestep).}
    \label{fig:log2trace}
    \vspace{-0.2cm}
\end{figure}

\begin{algorithm}[t]
  \caption{ConstructTraces}
  \label{alg:trace-construction}
  \KwIn{Logs $\Pi$, Variables $S$, Functions $F$, Predicates $P$, Rankings $\texttt{rank}$}
  \KwOut{TSL$_f$ traces $\{\sigma_1, \ldots, \sigma_m\}$}
  result $\gets [\ ]$\;
  \ForEach{$\pi \in \Pi$}{
    $\sigma \gets [\ ]$\;
    \ForEach{$i \in [0,|\pi|-1)$}{
        $\sigma_i \gets \{\}$\;
        \ForEach{$v \in S$}{
            \ForEach{$(f, \vec{w}) \in \texttt{rank}[v]$}{
                \If{$f(\texttt{val}^\pi_{i}(\vec{w})) = \texttt{val}^\pi_{i+1}(v)$}{
                    $\sigma_i \gets \sigma_i \cup \{\update{v}{f(w)}\}$; \textbf{break}\;
                }
            }
        }
        \ForEach{$p \in P$}{
            \lIf{$p(\texttt{val}^\pi_{i}(\vec{w}))$}{$\sigma_i \gets \sigma_i \cup \{\texttt{p(w)}\}$}
        }
    }
    $\sigma_{|\pi|-1} \gets \{\mathtt{END}\}$ \tcp*{only predicates and \texttt{END} at last timestep}
    \ForEach{$p \in P$}{
        \lIf{$p(\texttt{val}^\pi_{|\pi|-1}(\vec{w}))$}{$\sigma_{|\pi|-1} \gets \sigma_{|\pi|-1} \cup \{\texttt{p(w)}\}$}
    }
    result $\gets \sigma ::$ result\;
  }
  \Return{result}
  \end{algorithm}

\subsection{Lifting Well-Formed Traces}

Having discovered a function set $F$, we now construct well-formed TSL$_f$ traces. 
This requires converting raw log sequences into traces $\sigma = \sigma_0 \sigma_1 \cdots \sigma_n$ where each $\sigma_i \subseteq \mathcal{T}_u \cup \mathcal{T}_p$ contains the update and predicate terms at position $i$. 
Recall that for a trace to be well-formed, we must determinize how variables evolves: each timestep $i < |\sigma|-1$ can have exactly one update for each variable.
A challenge arises when multiple valid updates exist.
Consider $y$ at time$=$4 of \cref{fig:frozenlake}. 
Supposing $\mathtt{(+1)} \in F$, $[y \leftarrow y]_{t_4}$ is a valid update, as is $[y \leftarrow x+1]_{t_4}$.
While we know the ``correct assignment" to be the former, this is not explicit.

We remedy this by adopting a selection strategy that determinizes a single interpretation, ranking function applications by frequency. For each variable $v$, we rank function-input pairs by how frequently across all traces they are valid updates for $v$. 
If multiple updates are valid at a timestep, we select the highest-ranked option.
This score is aggregated across all logs presented. 
Algorithm~\ref{alg:trace-construction} formalizes this process, with Algorithm~\ref{alg:rankings} presenting the ranking mechanism. 

  \begin{algorithm}[t]
  \caption{ComputeRankings}
  \label{alg:rankings}
  \KwIn{Logs $\Pi$, Variables $S$, functions $F$}
  \KwOut{Rankings $\texttt{rank}$}
  \ForEach{$v \in S$}{
    $\texttt{rank}[v] \gets \{\}$\;
    \ForEach{$\pi \in \Pi$}{
      \ForEach{$t \in [0,|\pi|)$}{
          \ForEach{$(f, \vec{w}) \in F \times \text{InputCombinations}(v)$}{
                \If{$f(\texttt{val}^\pi_{t}(\vec{w})) = \texttt{val}^\pi_{t+1}(v)$}{
                    $\texttt{rank}[v][(f, \vec{w})] \mathrel{++}$\;
                }
          }
      }
    }
    $\texttt{sort}(\texttt{rank}[v], \lambda x. y.x > y)$ \tcp*{sort \texttt{rank[v]} descending}
  }
  \Return{$\texttt{rank}$}
  \end{algorithm}

While functions need to be discovered due to the infinitely many candidate configurations, predicates are tied to data types. Integers, for example, have equality \texttt{(=)} and ordering \texttt{(<)} defined. When constructing well-formed traces, we apply the predicates between all variables of the same type at every timestep.


\vspace{-0.2cm}

\subsection{Learning Specifications}

With well-formed TSL$_f$ traces, we now have a Booleanized interpretation of system behavior. This allows us to call specification mining solvers such as Bolt~\cite{bathie2025bolt},
treating update and predicate terms as atomic propositions. Bolt is a state-of-the-art mining tool that conducts a two-phased bounded search: temporal formulas are enumerated up to a bounded formula size, at which point Boolean compositions are evaluated as an instance of the Boolean Subset Cover problem. 


For finding a discriminant specification, this is sufficient. However, for synthesizing a \emph{reactive} program, a minimal discriminative formula is often an underspecifation. On \textsc{FrozenLake} for example, we would frequently only mine the goal condition ($\mathbf{F}\ \texttt{(=) p g}$) since this is the smallest discriminative formula. Yet synthesizing a controller from this specification yields an underspecified agent: the formula says nothing about reactive conditions (i.e. always avoid holes).

What we require is 
a specification that captures both what the agent must achieve and how it must react. This motivates decomposing the mining problem into two components: a \emph{liveness} condition $\mathbf{F}\psi$, encoding a goal to eventually reach, and a \emph{safety} condition $\mathbf{G}\,\varphi$, encoding invariants that always hold. Together, these form a specification $\mathbf{G}\varphi \land \mathbf{F}\psi$ that 
synthesize a winning controller. This structure aligns with Generalized Reactivity of Rank 1 specifications~\cite{bloem2012synthesis}.

We extend Bolt to mine safety and liveness specifications.
For liveness, we conduct the same bounded search but evaluate candidate formulas $\psi$ as $\mathbf{F}\psi$, seeking the smallest discriminative formula.
For safety mining, we evaluate candidates $\varphi$ wrapped as $\mathbf{G}\varphi$, additionally restricting the search grammar to Boolean connectives ($\land$, $\lor$, $\neg$, $\rightarrow$, $\leftrightarrow$) and next ($\mathbf{X}$) to ensure purely invariant properties. The search terminates upon finding the smallest formula of each type.

%% file: sections/6-evaluation.tex
\section{Evaluation}\label{sec:eval}

  \begin{figure}[t]
      \centering
      \includegraphics[width=1\linewidth]{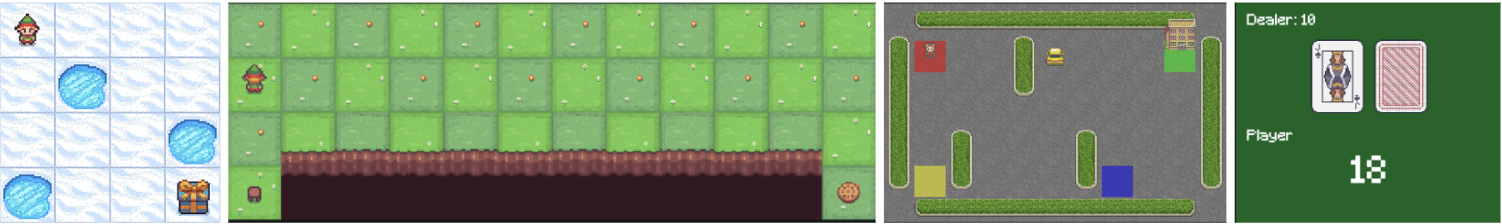}
      \caption{Full OpenAI-Gym ToyText~\cite{openai2016gym} suite. We evaluate on generalized instances. Left to right: \textsc{FrozenLake}, \textsc{CliffWalking}, \textsc{Taxi}, and \textsc{Blackjack}.}
      \label{fig:toytext}
      \vspace{-0.4cm}
  \end{figure}

We evaluate TSL$_f$ specification mining as a \emph{symbolic policy learning framework} across the full suite of environments from the OpenAI-Gym ToyText benchmarks~\cite{openai2016gym}. Mining specifications from traces, we synthesize reactive controllers which we deploy on unseen environment configurations. Trace examples and generation procedures are presented in \Cref{app:traces}. We collate total instances where action mechanics (legal moves) are respected and the objective is achieved. This quantifiably assesses our procedure's ability to discover both the action functions and the objective mechanics of these environments. 

\vspace{0.15em}
\noindent\textbf{Baseline Methods}
  We benchmark against \emph{passive} imitation learning methods.
  These baselines map to the same tasks, but their learning paradigms differ from ours as they learn localized policies from state-action samples. We also note that these methods learn exclusively from
  positive demonstrations, whereas our approach learns to discriminative specifications from positive and negative traces.
  
  \begin{itemize}
    \item \textbf{Stochastic Mealy Machine Learning (Alergia):} Alergia~\cite{carrasco1994learning} is a classical algorithm for learning probabilistic automata from positive examples by merging statistically compatible states in a prefix tree acceptor. 
    We require stochastic automata learning as our demonstrations can exhibit different actions from identical states (multiple valid paths exist to the goal). We use the AALp Alergia implementation~\cite{muskardin2022aalpy}
    We sample actions according to transition probabilities; for unseen states, we select actions uniformly at random.
      \item \textbf{Behavioral Cloning (NN):} Behavioral Cloning represents the neural approach to imitation learning~\cite{pomerleau1991efficient,hussein2017imitation}.
      with applications to game-playing~\cite{hester2018deep} and robotics~\cite{zhang2018deep}. We construct a feedforward neural network with two hidden layers (64 units each, ReLU activations) trained via supervised learning on state-action pairs. Training uses cross-entropy loss, Adam optimizer ($\alpha=0.001$), and early stopping on a 20\% validation split. 

      \item \textbf{Decision Tree (CART):} A CART decision tree~\cite{breiman1984cart} trained on state-action pairs.
      We use Gini impurity splitting with cost-complexity pruning. 

      \item \textbf{Bit-Blasting (LTL-BB):} To isolate the impact of TSL$_f$ as an expressive abstraction mechanism, we compare against bit-blast LTL$_f$ mining, encoding integer variables as Boolean bits which we treat as atomic propositions.
  \end{itemize}
  A distinction between these baselines and ours is that TSL$_f$ mining discovers action semantics, whereas these baselines assume
   provided action labels. By framing the task as ``given a state, select among known actions,'' these baselines receive additional supervision that our
  approach must infer from raw traces.

  \vspace{0.5em}
  \noindent\textbf{Strategy Synthesis and Deployment}
  Given a mined TSL$_f$ specification $\varphi_n$, we synthesize a reactive controller using Issy~\cite{heim2025issy}.
  Issy synthesizes $\varphi_n$ into a finite-state transducer that maps environment observations to actions.   We deploy this learned strategy on unseen test configurations. For grid-world games, we embed discovered update terms and mined liveness-safety specifications in a barebones skeleton template which specifies starting coordinates and grid bounds. A similar template is constructed for Blackjack. Templates are presented in \Cref{app:synt-templates}. 


\begin{figure}[t]
    \centering
    \includegraphics[width=0.9\linewidth]{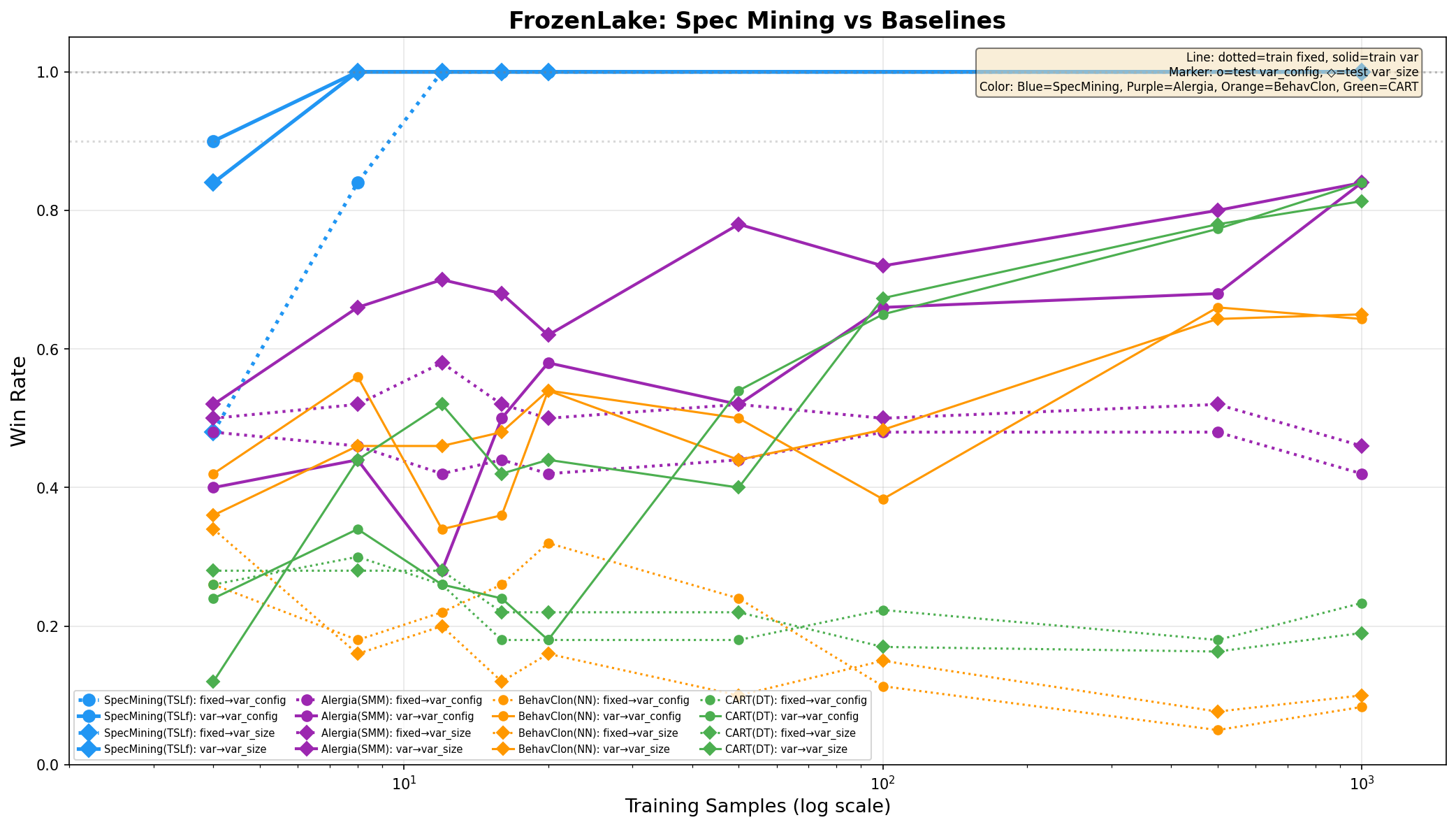}
    \caption{\textsc{FrozenLake} varied configuration results. Specification mining is more sample-efficient and generalizes better than imitation learning baselines.}
    \label{fig:frozen-lake-results}
    \vspace{-0.5cm}
\end{figure}

\vspace{0.5em}
\noindent\textbf{Results} We report results across the ToyText suite, evaluating generalization to unseen configurations, sample-efficiency, and performance of temporal-relational abstraction against localized state-action learning. For each environment, we compare against the learning baselines and examine specifications produced by LTL-BB. Statistics on learning and synthesis times are reported in \Cref{app:timing}

In \textsc{FrozenLake}, a player navigates a grid from start to goal without stepping on holes. The fixed instance is a $4\times4$ grid with holes at $(1,1)$, $(0,3)$, and $(3,2)$, and the goal at $(3,3)$. We evaluate generalization to unseen boards: \texttt{var\_conf} randomizes goal/hole positions on the grid, and \texttt{var\_size} additionally varies dimensions from $3\times3$ to $5\times5$.
TSL$_f$ mines from $n$ positive and $n$ negative traces (total $2n$), while imitation baselines learn from $n$ positive traces. 

\Cref{tab:frozen_lake_combined} reports win rates over 50 test instances under the two generalization settings. Specification Mining with TSL$_f$ achieves 100\% accuracy with 24 traces (12 positive, 12 negative) across all configurations, perfectly learning the environment policy of reaching the goal state while avoiding the holes. Baselines trained on varied configurations plateau between 70--85\% with 1000 samples; when trained on fixed configurations, none exceed 50\%.

The gap reflects what each method learns. Imitation baselines encode state-action mappings: a policy trained with holes at $(1,1), (0,3), (3,2)$ memorizes those avoidance maneuvers, with no transferability When hole positions change. TSL$_f$ mines relational specifications over predicates comparing player and hole coordinates. \Cref{tab:mined-specs-comparison} shows the mined safety formula (the biconditional encodes mutual exclusion; it can only be satisfied if the player visits none of the holes). This transfers to novel placements as the specification is defined over relations. Bit-blasting produces formulas such as $\mathbf{F}(x_{b1})$ (``eventually bit 1 of $x$ is true"). These are syntactically valid LTL$_f$ but lose any useful semantic interpretations.

\begin{table}[t]
\tiny
\centering
\resizebox{\textwidth}{!}{%
\begin{tabular}{lllccccccccc}
\toprule
\textbf{Method} & \textbf{Train} & \textbf{Test} & \textbf{4} & \textbf{8} & \textbf{12} & \textbf{16} & \textbf{20} & \textbf{50} & \textbf{100} & \textbf{500} & \textbf{1000} \\
\midrule
\multirow{4}{*}{\textbf{SpecMining(TSL$_f$)}} 
  & fixed & var\_conf & \cc{24} & \cc{42} & \cc{50} & \cc{50} & \cc{50} & \ccd & \ccd & \ccd & \ccd \\
  & var\_conf & var\_conf & \cc{45} & \cc{50} & \cc{50} & \cc{50} & \cc{50} & \ccd & \ccd & \ccd & \ccd \\
  & fixed & var\_size & \cc{24} & \cc{50} & \cc{50} & \cc{50} & \cc{50} & \ccd & \ccd & \ccd & \ccd \\
  & var\_size & var\_size & \cc{42} & \cc{50} & \cc{50} & \cc{50} & \cc{50} & \ccd & \ccd & \ccd & \ccd \\
\midrule
\multirow{4}{*}{Alergia(SMM)} 
  & fixed & var\_conf & \cc{24} & \cc{23} & \cc{21} & \cc{22} & \cc{21} & \cc{22} & \cc{24} & \cc{24} & \cc{21} \\
  & var\_conf & var\_conf & \cc{20} & \cc{22} & \cc{14} & \cc{25} & \cc{29} & \cc{26} & \cc{33} & \cc{34} & \cc{42} \\
  & fixed & var\_size & \cc{25} & \cc{26} & \cc{29} & \cc{26} & \cc{25} & \cc{26} & \cc{25} & \cc{26} & \cc{23} \\
  & var\_size & var\_size & \cc{26} & \cc{33} & \cc{35} & \cc{34} & \cc{31} & \cc{39} & \cc{36} & \cc{40} & \cc{42} \\
\midrule
\multirow{4}{*}{BehavClon(NN)} 
  & fixed & var\_conf & \cc{13} & \cc{9} & \cc{11} & \cc{13} & \cc{16} & \cc{12} & \cc{6} & \cc{3} & \cc{4} \\
  & var\_conf & var\_conf & \cc{21} & \cc{28} & \cc{17} & \cc{18} & \cc{27} & \cc{25} & \cc{19} & \cc{33} & \cc{32} \\
  & fixed & var\_size & \cc{17} & \cc{8} & \cc{10} & \cc{6} & \cc{8} & \cc{5} & \cc{8} & \cc{4} & \cc{5} \\
  & var\_size & var\_size & \cc{18} & \cc{23} & \cc{23} & \cc{24} & \cc{27} & \cc{22} & \cc{24} & \cc{32} & \cc{32} \\
\midrule
\multirow{4}{*}{CART(DT)} 
  & fixed & var\_conf & \cc{13} & \cc{15} & \cc{13} & \cc{9} & \cc{9} & \cc{9} & \cc{11} & \cc{9} & \cc{12} \\
  & var\_conf & var\_conf & \cc{12} & \cc{17} & \cc{13} & \cc{12} & \cc{9} & \cc{27} & \cc{33} & \cc{39} & \cc{42} \\
  & fixed & var\_size & \cc{14} & \cc{14} & \cc{14} & \cc{11} & \cc{11} & \cc{11} & \cc{8} & \cc{8} & \cc{9} \\
  & var\_size & var\_size & \cc{6} & \cc{22} & \cc{26} & \cc{21} & \cc{22} & \cc{20} & \cc{34} & \cc{39} & \cc{41} \\
\bottomrule
\end{tabular}%
}
\vspace{0.5em}
\caption{\textsc{FrozenLake} win rates (out of 50 test instances). \texttt{var\_conf}: fixed $4 \times 4$ board with varied hole/goal positions. \texttt{var\_size}: $3 \times 3$ to $5 \times 5$ boards with varied positions. Columns show number of training samples.}
\label{tab:frozen_lake_combined}
\vspace{-0.5cm}
\end{table}

  \begin{table*}[t]
\centering
\tiny
\resizebox{\textwidth}{!}{
\begin{tabular}{@{}l p{3.5cm} p{7cm}@{}}
\toprule
\textbf{Method} & \textbf{Liveness} & \textbf{Safety} \\
\midrule
\textbf{TSL$_f$}
& $\mathbf{F}\,((\texttt{eq}\,x\,\texttt{goalx}) \land (\texttt{eq}\,y\,\texttt{goaly}))$
&  $\mathbf{G}\, \neg ((\texttt{eq}\,x\,\texttt{hole1x}) \land (\texttt{eq}\,y\,\texttt{hole1y}))$ \\
& & $\mathbf{G}\,(((\texttt{eq}\,x\,\texttt{hole0x}) \land (\texttt{eq}\,y\,\texttt{hole0y})) \leftrightarrow (((\texttt{eq}\,x\,\texttt{hole1x}) \land (\texttt{eq}\,y\,\texttt{hole1y})) ))$ \\
& & $\mathbf{G}\,(((\texttt{eq}\,x\,\texttt{hole2x}) \land (\texttt{eq}\,y\,\texttt{hole2y})) \leftrightarrow (((\texttt{eq}\,x\,\texttt{hole0x}) \land (\texttt{eq}\,y\,\texttt{hole0y})) \lor ((\texttt{eq}\,x\,\texttt{hole1x}) \land (\texttt{eq}\,y\,\texttt{hole1y}))))$ \\
\midrule
LTL-BB
& $\mathbf{F}\,(\texttt{x\_b1})$
& $\mathbf{G}\,((\mathbf{X}\,\texttt{goalX\_b2}) \to (\texttt{x\_b1}))$ \\
& $\mathbf{F}\,((\texttt{y\_b0}) \land (\mathbf{X}\,\texttt{x\_b1}))$
& $\mathbf{G}\,((\texttt{y\_b0}) \to ((\texttt{y\_b1}) \leftrightarrow (\mathbf{X}\,\texttt{y\_b0})))$ \\
& $\mathbf{F}\,(\mathbf{G}\,((\texttt{y\_b0}) \land (\texttt{y\_b1})))$
& $\mathbf{G}\,((\mathbf{X}\,\texttt{goalX\_b2}) \to ((\texttt{y\_b0}) \land (\texttt{y\_b1})))$ \\
& $\mathbf{F}\,((\texttt{goalY\_b0}) \leftrightarrow (\texttt{hole1Y\_b1}))$
& $\mathbf{G}\,((\texttt{goalY\_b0}) \leftrightarrow (\texttt{hole1Y\_b1}))$ \\
\bottomrule
\end{tabular}
}
\vspace{0.5em}
\caption{Unique specifications found by TSL$_f$ and LTL-BB on \textsc{FrozenLake}.}
\label{tab:mined-specs-comparison}
\vspace{-1cm}
\end{table*}

\begin{table}[t]
\tiny
\centering
\resizebox{\textwidth}{!}{%
\begin{tabular}{lllcccccccccc}
\toprule
\textbf{Method} & \textbf{Train} & \textbf{Test} & \textbf{4} & \textbf{8} & \textbf{12} & \textbf{16} & \textbf{20} & \textbf{50} & \textbf{100} & \textbf{200} & \textbf{500} & \textbf{1000} \\
\midrule
\multirow{4}{*}{\textbf{SpecMining(TSL$_f$)}}
  & fixed & var\_size & \cc{17} & \cc{17} & \cc{17} & \cc{17} & \cc{17} & \cc{-} & \cc{-} & \cc{-} & \cc{-} & \cc{-} \\
  & var\_size & var\_size & \cc{50} & \cc{50} & \cc{50} & \cc{50} & \cc{50} & \cc{-} & \cc{-} & \cc{-} & \cc{-} & \cc{-} \\
  & var\_mov\_fixed & var\_size & \cc{17} & \cc{17} & \cc{17} & \cc{17} & \cc{17} & \cc{-} & \cc{-} & \cc{-} & \cc{-} & \cc{-} \\
  & var\_mov\_size & var\_size & \cc{0} & \cc{50} & \cc{50} & \cc{50} & \cc{50} & \cc{-} & \cc{-} & \cc{-} & \cc{-} & \cc{-} \\
\midrule
\multirow{4}{*}{Alergia(SMM)}
  & fixed & var\_size & \cc{19} & \cc{21} & \cc{22} & \cc{22} & \cc{21} & \cc{20} & \cc{18} & \cc{19} & \cc{20} & \cc{20} \\
  & var\_size & var\_size & \cc{50} & \cc{50} & \cc{50} & \cc{50} & \cc{50} & \cc{50} & \cc{50} & \cc{50} & \cc{50} & \cc{50} \\
  & var\_mov\_fixed & var\_size & \cc{22} & \cc{22} & \cc{24} & \cc{18} & \cc{18} & \cc{18} & \cc{14} & \cc{3} & \cc{0} & \cc{0} \\
  & var\_mov\_size & var\_size & \cc{20} & \cc{20} & \cc{18} & \cc{25} & \cc{22} & \cc{26} & \cc{18} & \cc{4} & \cc{0} & \cc{0} \\
\midrule
\multirow{4}{*}{BehavClon(NN)}
  & fixed & var\_size & \cc{5} & \cc{5} & \cc{2} & \cc{7} & \cc{7} & \cc{7} & \cc{7} & \cc{8} & \cc{8} & \cc{7} \\
  & var\_size & var\_size & \cc{20} & \cc{23} & \cc{43} & \cc{47} & \cc{50} & \cc{43} & \cc{47} & \cc{47} & \cc{23} & \cc{17} \\
  & var\_mov\_fixed & var\_size & \cc{7} & \cc{7} & \cc{7} & \cc{3} & \cc{7} & \cc{3} & \cc{3} & \cc{3} & \cc{0} & \cc{3} \\
  & var\_mov\_size & var\_size & \cc{17} & \cc{27} & \cc{43} & \cc{50} & \cc{50} & \cc{50} & \cc{30} & \cc{33} & \cc{23} & \cc{17} \\
\midrule
\multirow{4}{*}{CART(DT)}
  & fixed & var\_size & \cc{2} & \cc{2} & \cc{2} & \cc{2} & \cc{2} & \cc{2} & \cc{2} & \cc{2} & \cc{2} & \cc{2} \\
  & var\_size & var\_size & \cc{20} & \cc{27} & \cc{40} & \cc{50} & \cc{50} & \cc{47} & \cc{37} & \cc{37} & \cc{23} & \cc{20} \\
  & var\_mov\_fixed & var\_size & \cc{3} & \cc{3} & \cc{3} & \cc{0} & \cc{0} & \cc{3} & \cc{3} & \cc{3} & \cc{3} & \cc{3} \\
  & var\_mov\_size & var\_size & \cc{13} & \cc{27} & \cc{50} & \cc{50} & \cc{50} & \cc{47} & \cc{33} & \cc{37} & \cc{23} & \cc{20} \\
\bottomrule
\end{tabular}%
}
\vspace{0.5em}
\caption{\textsc{CliffWalking} win counts from training over fixed boards and varied board configurations. Test on 50 \texttt{var\_size} instances (cliff height $h\in[1,3]$,
board width $w\in[3,12]$, board length $l\in[4,6]$). }
\label{tab:cliff_walking_var_config}
\vspace{-0.5cm}
\end{table}


\begin{table*}[h]
\centering
\tiny
\resizebox{\textwidth}{!}{
\begin{tabular}{@{}l p{3.5cm} p{7cm}@{}}
\toprule
\textbf{Method} & \textbf{Liveness} & \textbf{Safety} \\
\midrule
\textbf{TSL$_f$}
& $\mathbf{F}\,((\texttt{eq}\,x\,\texttt{goalx}) \land (\texttt{eq}\,y\,\texttt{goaly}))$
&  $\mathbf{G}\,((\texttt{eq}\,x\,\texttt{goaly}) \lor ((\texttt{eq}\,y\,\texttt{goaly}) \to (\texttt{eq}\,x\,\texttt{goalx})))$ \\
& & $\mathbf{G}\,((\texttt{eq}\,x\,\texttt{goalx}) \lor ((\texttt{lt}\,y\,\texttt{cliffHeight}) \to (\texttt{eq}\,x\,\texttt{goalx})))$ \\
\midrule
LTL-BB
& $\mathbf{F}\,(\mathbin{\mathbf{G}} ((\texttt{x\_b0}) \land (\texttt{x\_b1})))$
& $\mathbf{G}\,((\mathbf{X}\,\texttt{cliffHeight\_b1}) \to ((\texttt{x\_b0}) \land (\texttt{x\_b1})))$ \\
& $\mathbf{F}\,(\mathbin{\mathbf{G}} ((\texttt{x\_b0}) \land (\texttt{x\_b3})))$
& $\mathbf{G}\,((\mathbf{X}\,\texttt{cliffHeight\_b1}) \to ((\texttt{x\_b0}) \land (\texttt{x\_b3})))$ \\
& $\mathbf{F}\,(\neg ((\mathbf{X}\,\texttt{x\_b3}) \to (\texttt{y\_b1})))$
& $\mathbf{G}\,((\mathbf{X}\,\texttt{cliffHeight\_b2}) \to (\texttt{x\_b3}))$ \\
\bottomrule
\end{tabular}
}
\vspace{0.5em}
\caption{Unique specifications found by TSL$_f$ and LTL-BB on \textsc{CliffWalking}.}
\label{tab:mined-specs-cliff-comparison}
\vspace{-1cm}
\end{table*}

In \textsc{CliffWalking}, a player starts at $(0,0)$ and must make his way to the rightmost $(w,0)$ point, circumnavigating a cliff. The fixed instance stipulates a $12\times4$ grid, with the player at $(0,0)$, the goal at $(11,0)$, and a cliff height of 1. We vary cliff height $h \in [1,3]$, board width $w \in [3,12]$, and board length $l \in [4,6]$. 

TSL$_f$ trained on varied configurations achieves perfect generalization. However, training on fixed configurations (cliff height 1) yields only 34\% on \texttt{var\_size} tests. The mined specification from fixed-configuration training (\Cref{tab:mined-specs-cliff-comparison}) states that ``either $x = 0$ or $y = 0$ implies $x = \mathit{goal}_x$". 
While this encapsulates winning strategies in the fixed case, increasing the cliff height extends the hazard region to $[0, h-1]$,
changing the semantic relationship between variables and the cliff. \Cref{tab:mined-specs-cliff-comparison} also shows imitation learning baselines learning perfect policies with \texttt{var\_size} training, but not fixed training.  Bit-blasted specifications for \textsc{CliffWalking} exhibit similar semantic limitations to \textsc{FrozenLake}.

We additionally evaluate a modified environment \texttt{var\_mov} with non-standard dynamics. Here, moving right updates $x \mapsto 2x+1$, moving up changes $y \mapsto y+2$. Function discovery recovers these transformations, constructing TSL$_f$ traces containing exactly the updates 
$[x \leftarrow x]$, $[x \leftarrow (x+x)+1]$, $[x \leftarrow x-1]$, $[y \leftarrow y]$, $[y \leftarrow y+2]$, $[y \leftarrow y-1]$. 
Synthesized controllers automatically respect these changed dynamics without modification to the mining or synthesis procedures.

It is worthwhile to note that in this setting, baseline performance decreases with additional samples. With movement functions $x \mapsto 2x+1$ and $y \mapsto y+2$, generating diverse traces requires extensive backtracking (each large forward step necessitates multiple smaller backward steps to produce distinct trajectories). Here, Alergia, which learns transition probabilities from state-action frequencies, increasingly observes backward moves for the same states, and thus learns this state-action mapping. 
Decision trees and neural networks exhibit similar degradation. Specification mining with TSL$_f$ is more robust, reasoning over temporal-relational global structures rather than local transition frequencies.

\begin{table*}[t]
\centering
\tiny
\resizebox{\textwidth}{!}{%
\begin{tabular}{lllcccccccccc}
\toprule
\textbf{Method} & \textbf{Train} & \textbf{Test} & \textbf{4} & \textbf{8} & \textbf{12} & \textbf{16} & \textbf{20} & \textbf{50} & \textbf{100} & \textbf{200} & \textbf{500} & \textbf{1000} \\
\midrule
\multirow{2}{*}{\textbf{SpecMining(TSL$_f$)}}
  & fixed & var\_pos & \cc{2} & \cc{31} & \cc{50} & \cc{50} & \cc{50} & \cc{-} & \cc{-} & \cc{-} & \cc{-} & \cc{-} \\
  & var\_pos & var\_pos & \cc{2} & \cc{50} & \cc{50} & \cc{50} & \cc{50} & \cc{-} & \cc{-} & \cc{-} & \cc{-} & \cc{-} \\
\midrule
\multirow{2}{*}{Alergia(SMM)}
  & fixed & var\_pos & \cc{14} & \cc{19} & \cc{5} & \cc{14} & \cc{17} & \cc{14} & \cc{17} & \cc{17} & \cc{23} & \cc{19} \\
  & var\_pos & var\_pos & \cc{27} & \cc{27} & \cc{23} & \cc{18} & \cc{23} & \cc{29} & \cc{32} & \cc{28} & \cc{30} & \cc{31} \\
\midrule
\multirow{2}{*}{BehavClon(NN)}
  & fixed & var\_pos & \cc{0} & \cc{0} & \cc{0} & \cc{1} & \cc{0} & \cc{0} & \cc{2} & \cc{1} & \cc{0} & \cc{0} \\
  & var\_pos & var\_pos & \cc{2} & \cc{4} & \cc{3} & \cc{12} & \cc{0} & \cc{9} & \cc{9} & \cc{14} & \cc{17} & \cc{14} \\
\midrule
\multirow{2}{*}{CART(DT)}
  & fixed & var\_pos & \cc{1} & \cc{1} & \cc{1} & \cc{2} & \cc{1} & \cc{1} & \cc{1} & \cc{1} & \cc{1} & \cc{1} \\
  & var\_pos & var\_pos & \cc{2} & \cc{0} & \cc{1} & \cc{2} & \cc{2} & \cc{2} & \cc{4} & \cc{13} & \cc{10} & \cc{13} \\
\midrule
\multirow{2}{*}{BehavClon(NN)$^*$}
  & fixed & var\_pos & \cc{5} & \cc{7} & \cc{2} & \cc{5} & \cc{4} & \cc{3} & \cc{2} & \cc{2} & \cc{3} & \cc{3} \\
  & var\_pos & var\_pos & \cc{3} & \cc{16} & \cc{0} & \cc{20} & \cc{6} & \cc{15} & \cc{20} & \cc{17} & \cc{28} & \cc{24} \\
\midrule
\multirow{2}{*}{CART(DT)$^*$}
  & fixed & var\_pos & \cc{1} & \cc{2} & \cc{2} & \cc{1} & \cc{3} & \cc{1} & \cc{4} & \cc{4} & \cc{0} & \cc{0} \\
  & var\_pos & var\_pos & \cc{3} & \cc{10} & \cc{0} & \cc{13} & \cc{5} & \cc{17} & \cc{18} & \cc{12} & \cc{26} & \cc{24} \\
\bottomrule
\end{tabular}%
}
\vspace{0.5em}
\caption{\textsc{Taxi} win counts from training over fixed boards and varied board configurations. Test on 50 \texttt{var\_pos} instances, varying the wall, passenger, and destinations positions. $*$ implies a dual-stage system with two baselines chained.}
\label{tab:taxi_var_pos}
\vspace{-0.5cm}
\end{table*}

\begin{table*}[t]
\centering
\tiny
\resizebox{\textwidth}{!}{
\begin{tabular}{@{}l p{4.8cm} p{5.8cm}@{}}
\toprule
\textbf{Method} & \textbf{Liveness} & \textbf{Safety} \\
\midrule
\textbf{TSL$_f$}
& $\mathbf{F}\,(\texttt{eq}\,x\,\texttt{DEST\_x} \land \texttt{eq}\,y\,\texttt{DEST\_y})$
& $\mathbf{G}\,(\texttt{eq}\,x\,\texttt{YELLOW\_x} \land \texttt{eq}\,y\,\texttt{YELLOW\_y})$ \\
& $\mathbf{F}\,((\texttt{eq}\,x\,\texttt{RED\_X} \land \texttt{eq}\,y\,\texttt{RED\_Y}) \land \mathbf{F}\,(\texttt{eq}\,x\,\texttt{DEST\_X} \land \texttt{eq}\,y\,\texttt{DEST\_Y}))$
& $\mathbf{G}\,((\texttt{eq}\,x\,\texttt{RED\_X} \land \texttt{eq}\,y\,\texttt{RED\_Y}) \leftrightarrow (\texttt{eq}\,x\,\texttt{YELLOW\_X} \land \texttt{eq}\,y\,\texttt{YELLOW\_Y}))$ \\
\midrule
LTL-BB
& $\mathbf{F}\,((\texttt{x\_b0}) \leftrightarrow (\texttt{RED\_x\_b1}))$
& $\mathbf{G}\,((\texttt{x\_b0}) \leftrightarrow (\texttt{RED\_x\_b1}))$ \\
& $\mathbf{F}\,((\texttt{DEST\_x\_b1}) \leftrightarrow (((\texttt{DEST\_y\_b0}) \leftrightarrow (\texttt{RED\_y\_b1})) \to (\texttt{DEST\_x\_b0})))$
& $\mathbf{G}\,((\texttt{DEST\_x\_b1}) \leftrightarrow (((\texttt{DEST\_y\_b0}) \leftrightarrow (\texttt{RED\_y\_b1})) \to (\texttt{DEST\_x\_b0})))$ \\
\bottomrule
\end{tabular}
}
\vspace{0.5em}
\caption{Unique specifications mined by TSL$_f$ and LTL-BB on \textsc{Taxi}. LTL-BB mining times out for $n \geq 12$.}
\label{tab:taxi-mined-specs-comparison}
\vspace{-1cm}
\end{table*}

In \textsc{Taxi}, the agent must pick up a passenger at one of four designated locations and drop them off at a destination in a $5\times5$ grid with walls, losing if they drop the passenger at the wrong destination. The fixed case has the agent starting at $(2,2)$, the passenger at $(0,0)$, the destination at $(4,0)$, and incorrect drop-offs at $(0,4)$ and $(3,4)$.  \Cref{tab:taxi_var_pos} reports results under \texttt{var\_pos} generalization, varying start position, wall placement, passenger location, and destination.

\begin{table*}[t]
\centering
\tiny
\resizebox{\textwidth}{!}{%
\begin{tabular}{@{}llccccc@{}}
\toprule
\textbf{Method} & \textbf{Strategy} & \textbf{4} & \textbf{8} & \textbf{12} & \textbf{16} & \textbf{20} \\
\midrule
\multirow{3}{*}{\textbf{SpecMining(TSL$_f$)}}
    & Threshold    & \ccp{100}{44} & \ccp{100}{44} & \ccp{100}{44} & \ccp{100}{44} & \ccp{100}{44} \\
    & Conservative & \ccp{100}{49} & \ccp{80}{44} & \ccp{100}{49} & \ccp{100}{49} & \ccp{100}{49} \\
    & Basic        & \ccp{88}{44} & \ccp{88}{44} & \ccp{100}{50} & \ccp{100}{50} & \ccp{100}{50} \\
\midrule
\multirow{3}{*}{Alergia(SMM)}
    & Threshold    & \ccp{100}{44} & \ccp{100}{44} & \ccp{100}{44} & \ccp{100}{44} & \ccp{100}{44} \\
    & Conservative & \ccp{100}{49} & \ccp{100}{49} & \ccp{100}{49} & \ccp{100}{49} & \ccp{100}{49} \\
    & Basic        & \ccp{88}{44} & \ccp{100}{50} & \ccp{100}{50} & \ccp{100}{50} & \ccp{100}{50} \\
\midrule
\multirow{3}{*}{BehavClon(NN)}
    & Threshold    & \ccp{92}{36} & \ccp{100}{44} & \ccp{88}{43} & \ccp{100}{44} & \ccp{92}{36} \\
    & Conservative & \ccp{58}{46} & \ccp{84}{46} & \ccp{94}{46} & \ccp{92}{50} & \ccp{90}{46} \\
    & Basic        & \ccp{88}{44} & \ccp{68}{47} & \ccp{98}{46} & \ccp{90}{48} & \ccp{88}{44} \\
\midrule
\multirow{3}{*}{CART(DT)}
    & Threshold    & \ccp{92}{36} & \ccp{100}{44} & \ccp{100}{44} & \ccp{100}{44} & \ccp{100}{44} \\
    & Conservative & \ccp{58}{46} & \ccp{94}{46} & \ccp{94}{46} & \ccp{94}{46} & \ccp{94}{46} \\
    & Basic        & \ccp{88}{44} & \ccp{58}{48} & \ccp{86}{44} & \ccp{90}{46} & \ccp{98}{46} \\
\bottomrule
\end{tabular}%
}
\vspace{0.5em}
\caption{Strategy adherence (win rate) on \textsc{Blackjack} for different training set sizes on 50 instances.
As strategy becomes more complex, win rates improve.}
\label{tab:blackjack-results}
\vspace{-0.6cm}
\end{table*}

\begin{table*}[t]
\centering
\tiny
\resizebox{\textwidth}{!}{
\begin{tabular}{@{}l l p{9cm}@{}}
\toprule
\textbf{Method} & \textbf{Strategy} & \textbf{Safety Specification} \\
\midrule
\multirow{6}{*}{TSL$_f$}
& Threshold
& $\mathbf{G}\, \neg ((\texttt{lt}\,\texttt{c}\,\texttt{sThresh}) \leftrightarrow (\mathbf{X}\,\texttt{stood}))$ \\[0.5em]
& \multirow{3}{*}{Conservative}
& $\mathbf{G}\, ((\texttt{lt}\,\texttt{c}\,\texttt{sThresh}) \to (\mathbf{X}\, ((\texttt{lt}\,\texttt{c}\,\texttt{sThresh}) \leftrightarrow (((\texttt{gte}\,\texttt{dealer}\,\texttt{lo}) \land (\texttt{lte}\,\texttt{dealer}\,\texttt{med})) \lor (\texttt{lt}\,\texttt{c}\,\texttt{sWeak})))))$ \\
& & $\mathbf{G}\, \neg ((\mathbf{X}\,\texttt{stood}) \leftrightarrow ((\texttt{lt}\,\texttt{c}\,\texttt{sThresh}) \land (((\texttt{gte}\,\texttt{dealer}\,\texttt{lo}) \land (\texttt{lte}\,\texttt{dealer}\,\texttt{med})) \to (\texttt{lt}\,\texttt{c}\,\texttt{sWeak}))))$ \\[0.5em]
& \multirow{3}{*}{Basic}
& $\mathbf{G}\, ((\texttt{lt}\,\texttt{c}\,\texttt{sWeak}) \lor (\mathbf{X}\, ((\texttt{eq}\,\texttt{c}\,\texttt{sWeak}) \leftrightarrow [\texttt{c} \leftarrow \texttt{c}])))$ \\
& & $\mathbf{G}\, \neg ((\mathbf{X}\,\texttt{stood}) \leftrightarrow ((\texttt{lt}\,\texttt{c}\,\texttt{sThresh}) \land (((\texttt{gte}\,\texttt{dealer}\,\texttt{lo}) \land (\texttt{lte}\,\texttt{dealer}\,\texttt{med})) \to (\texttt{lt}\,\texttt{c}\,\texttt{sWeak}))))$ \\
\midrule
\multirow{3}{*}{LTL-BB}
& Threshold
& $\mathbf{G}\, ((\texttt{c\_b4}) \to (\mathbf{X}\,\texttt{stood}))$ \\[0.5em]
& Conservative
& $\mathbf{G}\, ((\texttt{dealer\_b2}) \leftrightarrow (\texttt{c\_b3}))$ \\[0.5em]
& Basic
& $\mathbf{G}\, ((\texttt{dealer\_b1}) \leftrightarrow (\mathbf{X}\,\texttt{dealer\_b1}))$ \\
\bottomrule
\end{tabular}
}
\vspace{0.5em}
\caption{Specifications mined by TSL$_f$ and LTL-BB on \textsc{Blackjack}.\\
\textbf{Legend:} \texttt{c}$=$player card total, \texttt{lo}$=$2, \texttt{med}$=$6, \texttt{sThresh}$=$17, \texttt{sWeak}$=$12.}
\label{tab:blackjack-mined-specs}
\vspace{-1cm}
\end{table*}

This environment is interesting due to its temporally ordered goal condition: the agent must first navigate to the passenger and then to the goal.
Imitation baselines do not exceed 32\% with 1000 samples, whether trained on fixed or varied configurations. This failure is representational: these methods learn stateless mappings from observations to actions, which cannot encode ``go to passenger, then go to goal'' as a policy. 
To isolate this effect, we evaluate two-stage baselines (NN$^*$, DT$^*$) that chain separate models for each phase. One is trained on traces to the passenger, another on traces from passenger to destination. Performance improves to 56\%, suggesting an inability to capture temporal structure.  With 24 total traces, 
TSL$_f$ mines the sequential objective directly
: eventually reach the passenger, and from then eventually reach the destination (see \Cref{tab:taxi-mined-specs-comparison}). 
LTL-BB times out for $n \geq 12$ traces, otherwise deriving spurious bit equivalences.

\textsc{Blackjack}, the last of the Gymnasium environments, differs fundamentally from the preceding grid worlds. There are no spatial coordinates, no movement functions to discover, game outcomes are stochastic, and episodes last at most three moves. It is not a natural fit for temporal-relational reasoning: decisions depend on the local state (current hand, dealer card), with no sequential structure to exploit. Despite this, TSL$_f$ shows an ability to mine strategies. 

Since game outcomes are non-deterministic in \textsc{Blackjack}, we discriminate traces by \emph{strategy adherence}, separating by whether gameplay respects three strategies of increasing complexity: \emph{Threshold} (stand if hand $\geq 17$), \emph{Conservative} (incorporate dealer card into the decision), and \emph{Basic} (full basic strategy). For formal definitions of strategies, see \Cref{app:blackjack}. 
Because no deterministic goal discriminates the traces, there is no liveness condition to mine (liveness mining times out). We thus only mine safety specifications that constrain when to stand.
Moreover, we note that TSL$_f$ mines predicates relationally over variables, but the Blackjack strategies reference absolute thresholds (e.g., 17). To enable mining, we add constants to traces (see \cref{tab:blackjack-mined-specs} legend). Providing these allows the miner to express predicates against constant bounds.
This reflects a limitation of purely relational mining: TSL$_f$ cannot natively discover that 17 is a salient threshold from traces alone. Parametric identification techniques~\cite{asarin2011parametric} that infer predicate boundaries from data present a direction for future work.

\Cref{tab:blackjack-results} reports strategy adherence over 50 test instances. Alergia achieves optimal performance, an expected result as Blackjack decisions are memoryless and naturally represented by stochastic automata mappings. TSL$_f$ achieves 100\% adherence with 8--24 traces depending on strategy complexity, correctly learning all strategies. 
Bit-blasted specifications occasionally approximate correctness: $\mathbf{G}(c_{b4} \rightarrow (\mathbf{X}\,\mathtt{stood}))$ captures that bit 4 (hand $\geq 16$) relates to standing. 

%% file: sections/2-related.tex
\section{Related Works}\label{sec:related}



\noindent\textbf{Temporal Logics}
Linear Temporal Logic (LTL)~\cite{pnueli1977temporal} and its finite-trace variant LTL$_f$~\cite{giacomo2013ltlf} remain the standard formalisms for temporal specification, underpinning work in planning~\cite{bacchus2000using, aminof2025multiagent}, process mining~\cite{komatsu2024process}, and runtime verification~\cite{clarke1999model}. However, these logics are propositional, encoding events as Boolean atoms. Temporal Stream Logic (TSL)~\cite{santolucito2019tsl} separates control from data, enabling specifications over uninterpreted functions and predicates; TSL-MT~\cite{bernd2022tslmt} grounds these in background theories. Alternative data-aware extensions include LTL$_f$ Modulo-Theories ~\cite{geatti2022ltlfmt}, which supports monotonic counters, and data-LTL$_f$~\cite{gianola2025smt}, an SMT-backed variant with variable conditionals. Signal Temporal Logic (STL)~\cite{bartocci2022survey} is another interesting paradigm that reasons over continuous signals with quantitative semantics. Our work presents TSL$_f$, a finite-trace interpretation of TSL.

\vspace{0.5em}

\noindent\textbf{Specification Mining}
Specification mining extracts temporal properties from traces. Texada~\cite{lemieux2015general} performs template-based LTL mining; Bolt~\cite{bathie2025bolt} and Scarlet~\cite{neider2018learning} scale LTL$_f$ mining through set-cover formulations. These operate over Boolean abstractions, requiring hand-crafted atomic propositions. In the continuous domain, STL mining~\cite{bartocci2022survey} has received attention, with parametric identification techniques~\cite{asarin2011parametric} learning thresholds from signals. 

\vspace{0.5em}
\noindent\textbf{Syntax-Guided Synthesis}
SyGuS~\cite{alur2013syntax} techniques synthesize programs from semantic constraints and syntactic grammars, enabling Programming-by-Example techniques such as FlashFill~\cite{gulwani2011flashfill}. Modern solvers like CVC5~\cite{barbosa2022cvc5} support theories including linear arithmetic and bitvectors. SyGuS has been applied to motion planning~\cite{chasins2016using}. Das et al.~\cite{das2023combining} combine functional synthesis with automata learning to discover reactive programs from grid-world observations. Our work uses SyGuS as the primitive for function discovery, then lifts discovered functions into temporal specifications, melding syntax-guided and reactive synthesis.

%% file: sections/7-conclusion.tex
\section{Conclusion}

In this paper, we presented a framework for mining temporal specifications over rich datatypes from execution traces. By leveraging Syntax-Guided Synthesis to discover functions that explain variable evolutions, and formalizing TSL$_f$ as a finite-trace interpretation of Temporal Stream Logic, our approach enables specification mining that captures both data transformations and temporal structure. Our evaluation on the OpenAI-Gymnasium ToyText benchmarks demonstrates that reactive controllers synthesized from mined specifications achieve perfect accuracy on unseen configurations while requiring significantly fewer samples than passive learning baselines. As future work, a natural extension is closing the loop toward fully symbolic reinforcement learning, where an agent actively generates traces, mines specifications from experience, and iteratively refines its behavior through formal synthesis.

%% file: appendix/proofs.tex
\section{Thm.1 proof. TSL$_f$ is strictly less expressive than TSL.}\label{sec:thm1}
  
  \begin{proof}
  The inclusion TSL$_f$ $\subseteq$ TSL is immediate, since TSL$_f$ uses the same syntax as TSL with finite trace semantics. We show the inclusion is strict by exhibiting a TSL specification that is realizable in TSL but unrealizable in TSL$_f$.

  Consider the following TSL specification:
  $$
  \mathbf{G} \left( \mathbf{F} \, \update{x}{y} \land \mathbf{F} \, \update{x}{z} \right).
  $$

  This formula asserts that at every position, there will eventually be an update of $x$ from $y$ and an update of $x$ from $z$.

  \textbf{Realizable in TSL (infinite traces):}
  A system can satisfy this by alternating between the two updates indefinitely. For instance, the infinite computation pattern:
  $$\update{x}{y}, \update{x}{z}, \update{x}{y}, \update{x}{z}, \ldots$$
  satisfies the formula at every position, since from any position there is always a future position with each update.

  \textbf{Unrealizable in TSL$_f$ (finite traces):}
  Suppose for contradiction that some well-formed finite trace $\sigma$ of length $n$ satisfies the formula. Then:
  \begin{enumerate}
    \item By the semantics of $\mathbf{G}$, the formula must hold at every position, including the final position $n-1$:
    $$\sigma, n-1 \models \mathbf{F} \, \update{x}{y} \land \mathbf{F} \, \update{x}{z}$$

    \item By the strong semantics of $\mathbf{F}$ in TSL$_f$, where $\mathbf{F} \varphi \equiv \top \mathbin{\mathbf{U}} \varphi$, we have that $\mathbf{F} \, \update{x}{y}$ at position $n-1$ requires:
    $$\exists j \geq n-1. \, (j < n \land \sigma, j \models \update{x}{y})$$
    The only position $j$ satisfying $j \geq n-1$ and $j < n$ is $j = n-1$. Thus:
    $$\sigma, n-1 \models \update{x}{y}$$

    \item By symmetric reasoning, $\sigma, n-1 \models \update{x}{z}$.

    \item Therefore, both $\update{x}{y} \in \sigma_{n-1}$ and $\update{x}{z} \in \sigma_{n-1}$. This violates well-formedness.
  \end{enumerate}


  Thus, no well-formed finite trace can satisfy this specification, proving it is unrealizable in TSL$_f$. Since the specification is realizable in TSL but not in TSL$_f$, we conclude TSL$_f$ $\subsetneq$ TSL.
  \end{proof}

%% file: appendix/evalextra.tex
  \section{Trace Generation}\label{app:traces}

  All traces are recorded as sequences of state snapshots in \texttt{jsonl} format, where each state is a dictionary mapping variable names to integer values.
  For each environment and training size, we generate two sets of behavioral demonstrations.
  Positive traces are successful episodes that reach the goal state. For grid-based environments (\textsc{FrozenLake, CliffWalking, Taxi}), we use breadth-first search (BFS) to compute optimal paths, then add random detours to introduce trajectory diversity. Negative Traces are failed episodes generated by a mixture of random walks and deviations of positive trace BFS search. 
  For \textsc{Blackjack}, we use basic strategies (well-known deterministic policies based on hand totals). 
  Negative traces constitute traces that do not adhere to the strategy.

  \paragraph{Frozen Lake.}
  A winning trace navigating from start to goal, avoiding holes.

  \begin{lstlisting}[language=jsonabbr]
  {"x": 0, "y": 0, "goalx": 3, "goaly": 3, "h0x": 1, "h0y": 1, "h1x": 3, "h1y": 1, "h2x": 3, "h2y": 2}
  {"x": 1, "y": 0, "goalx": 3, "goaly": 3, "h0x": 1, "h0y": 1, "h1x": 3, "h1y": 1, "h2x": 3, "h2y": 2}
  {"x": 2, "y": 0, "goalx": 3, "goaly": 3, "h0x": 1, "h0y": 1, "h1x": 3, "h1y": 1, "h2x": 3, "h2y": 2}
  {"x": 2, "y": 1, "goalx": 3, "goaly": 3, "h0x": 1, "h0y": 1, "h1x": 3, "h1y": 1, "h2x": 3, "h2y": 2}
  {"x": 2, "y": 2, "goalx": 3, "goaly": 3, "h0x": 1, "h0y": 1, "h1x": 3, "h1y": 1, "h2x": 3, "h2y": 2}
  {"x": 2, "y": 3, "goalx": 3, "goaly": 3, "h0x": 1, "h0y": 1, "h1x": 3, "h1y": 1, "h2x": 3, "h2y": 2}
  {"x": 3, "y": 3, "goalx": 3, "goaly": 3, "h0x": 1, "h0y": 1, "h1x": 3, "h1y": 1, "h2x": 3, "h2y": 2}
  \end{lstlisting}

  \paragraph{Cliff Walking.}
  A winning trace navigating above the cliff to the goal.

  \begin{lstlisting}[language=jsonabbr]
  {"x": 0, "y": 0, "goalx": 11, "goaly": 0, "cliffXMin": 1, "cliffXMax": 10, "cliffHeight": 1}
  {"x": 0, "y": 1, "goalx": 11, "goaly": 0, "cliffXMin": 1, "cliffXMax": 10, "cliffHeight": 1}
  {"x": 1, "y": 1, "goalx": 11, "goaly": 0, "cliffXMin": 1, "cliffXMax": 10, "cliffHeight": 1}
  {"x": 2, "y": 1, "goalx": 11, "goaly": 0, "cliffXMin": 1, "cliffXMax": 10, "cliffHeight": 1}
  {"x": 3, "y": 1, "goalx": 11, "goaly": 0, "cliffXMin": 1, "cliffXMax": 10, "cliffHeight": 1}
  {"x": 4, "y": 1, "goalx": 11, "goaly": 0, "cliffXMin": 1, "cliffXMax": 10, "cliffHeight": 1}
  {"x": 5, "y": 1, "goalx": 11, "goaly": 0, "cliffXMin": 1, "cliffXMax": 10, "cliffHeight": 1}
  {"x": 6, "y": 1, "goalx": 11, "goaly": 0, "cliffXMin": 1, "cliffXMax": 10, "cliffHeight": 1}
  {"x": 7, "y": 1, "goalx": 11, "goaly": 0, "cliffXMin": 1, "cliffXMax": 10, "cliffHeight": 1}
  {"x": 8, "y": 1, "goalx": 11, "goaly": 0, "cliffXMin": 1, "cliffXMax": 10, "cliffHeight": 1}
  {"x": 9, "y": 1, "goalx": 11, "goaly": 0, "cliffXMin": 1, "cliffXMax": 10, "cliffHeight": 1}
  {"x": 10, "y": 1, "goalx": 11, "goaly": 0, "cliffXMin": 1, "cliffXMax": 10, "cliffHeight": 1}
  {"x": 11, "y": 1, "goalx": 11, "goaly": 0, "cliffXMin": 1, "cliffXMax": 10, "cliffHeight": 1}
  {"x": 11, "y": 0, "goalx": 11, "goaly": 0, "cliffXMin": 1, "cliffXMax": 10, "cliffHeight": 1}
  \end{lstlisting}

  \paragraph{Taxi.}
  A winning trace: taxi picks up passenger at RED, delivers to DEST.

  \begin{lstlisting}[language=jsonabbr]
  {"x": 2, "y": 2, "DEST_x": 4, "DEST_y": 0, "PASS_x": 0, "PASS_y": 0, "RED_x": 0, "RED_y": 0, "GREEN_x": 4, "GREEN_y": 0, "YELLOW_x": 0, "YELLOW_y": 4, "BLUE_x": 3, "BLUE_y":4}
  {"x": 1, "y": 2, "DEST_x": 4, "DEST_y": 0, "PASS_x": 0, "PASS_y": 0, "RED_x": 0, "RED_y": 0, "GREEN_x": 4, "GREEN_y": 0, "YELLOW_x": 0, "YELLOW_y": 4, "BLUE_x": 3, "BLUE_y":4}
  {"x": 0, "y": 2, "DEST_x": 4, "DEST_y": 0, "PASS_x": 0, "PASS_y": 0, "RED_x": 0, "RED_y": 0, "GREEN_x": 4, "GREEN_y": 0, "YELLOW_x": 0, "YELLOW_y": 4, "BLUE_x": 3, "BLUE_y":4}
  {"x": 0, "y": 1, "DEST_x": 4, "DEST_y": 0, "PASS_x": 0, "PASS_y": 0, "RED_x": 0, "RED_y": 0, "GREEN_x": 4, "GREEN_y": 0, "YELLOW_x": 0, "YELLOW_y": 4, "BLUE_x": 3, "BLUE_y":4}
  {"x": 0, "y": 0, "DEST_x": 4, "DEST_y": 0, "PASS_x": 0, "PASS_y": 0, "RED_x": 0, "RED_y": 0, "GREEN_x": 4, "GREEN_y": 0, "YELLOW_x": 0, "YELLOW_y": 4, "BLUE_x": 3, "BLUE_y":4}
  {"x": 1, "y": 0, "DEST_x": 4, "DEST_y": 0, "PASS_x": 0, "PASS_y": 0, "RED_x": 0, "RED_y": 0, "GREEN_x": 4, "GREEN_y": 0, "YELLOW_x": 0, "YELLOW_y": 4, "BLUE_x": 3, "BLUE_y":4}
  {"x": 2, "y": 0, "DEST_x": 4, "DEST_y": 0, "PASS_x": 0, "PASS_y": 0, "RED_x": 0, "RED_y": 0, "GREEN_x": 4, "GREEN_y": 0, "YELLOW_x": 0, "YELLOW_y": 4, "BLUE_x": 3, "BLUE_y":4}
  {"x": 3, "y": 0, "DEST_x": 4, "DEST_y": 0, "PASS_x": 0, "PASS_y": 0, "RED_x": 0, "RED_y": 0, "GREEN_x": 4, "GREEN_y": 0, "YELLOW_x": 0, "YELLOW_y": 4, "BLUE_x": 3, "BLUE_y":4}
  {"x": 4, "y": 0, "DEST_x": 4, "DEST_y": 0, "PASS_x": 0, "PASS_y": 0, "RED_x": 0, "RED_y": 0, "GREEN_x": 4, "GREEN_y": 0, "YELLOW_x": 0, "YELLOW_y": 4, "BLUE_x": 3, "BLUE_y":4}
  \end{lstlisting}

  \paragraph{Blackjack (Basic Strategy).}
  Hit at 11, stand at 15 vs weak dealer (showing 5).

  \begin{lstlisting}[language=jsonabbr]
  {"count": 11, "stood": false, "standThreshold": 17, "standVsWeakMin": 13, "isWeakDealer": true}
  {"count": 15, "stood": false, "standThreshold": 17, "standVsWeakMin": 13, "isWeakDealer": true}
  {"count": 15, "stood": true, "standThreshold": 17, "standVsWeakMin": 13, "isWeakDealer": true}
  \end{lstlisting}

    \section{Synthesis Templates}\label{app:synt-templates}

  \begin{lstlisting}[language=tsl,
      caption={\texttt{grid\_standard\_template.tsl}. \texttt{xMoves} and \texttt{yMoves} are disjunctions of update terms for \texttt{x} and \texttt{y} found in our TSL$_f$ mining process. \textsc{Taxi} template contains additional wall bounds.},
      basicstyle=\ttfamily\scriptsize,
      label={lst:grid.tsl},
      escapechar=|]

  inBounds = (gte x B_MIN) && (lte x B_MAX) && (gte y B_MIN) && (lte y B_MAX)
            // && !(eq x w0x && eq y w0y) && !(eq x w00x && eq y w00y) {Taxi}
            // && !(eq x w1x && eq y w1y) && !(eq x w11x && eq y w11y) {Taxi}
  xMoves = {discovered_x_updates}
  yMoves = {discovered_y_updates}
  assume {
      eq x START_X;
      eq y START_Y;
  }
  guarantee {
      G inBounds;                                       // Stay in bounds
      G ((xMoves && [y <- y]) || ([x <- x] && yMoves));     // Discovered updates
      {mined_tslf_specification};                       // Mined Objective
  }
  \end{lstlisting}

  Similarly for blackjack, the template sets bounds on potential card values, with the cards dealt as inputs to the system. The system controls \texttt{stood}:
  \begin{lstlisting}[language=tsl,
    caption={\texttt{blackjack.tsl}},
    basicstyle=\ttfamily\scriptsize,
    label={lst:blackjack.tsl}]
    
always assume {
    (gte handValue MIN_HAND) && (lte handValue MAX_HAND);
    (gte dealerCard MIN_DEALER) && (lte dealerCard MAX_DEALER);
}

guarantee {
    {mined_tslf_specification};
}
\end{lstlisting}

\section{Blackjack Strategies}\label{app:blackjack}

\begin{enumerate}
      \item \emph{Threshold Strategy.}
      \[
      \scriptsize
      \text{action}(c, d) =
      \begin{cases}
          \textsc{hit} & \text{if } c < 17 \\
          \textsc{stand} & \text{if } c \geq 17
      \end{cases}
      \]

      \item \emph{Conservative Strategy.}
      \[
      \scriptsize
      \text{action}(c, d) =
      \begin{cases}
          \textsc{stand} & \text{if } c \geq 17 \\
          \textsc{stand} & \text{if } c \geq 12 \land 2 \leq d \leq 6 \\
          \textsc{hit} & \text{if } c \leq 11 \\
          \textsc{hit} & \text{otherwise}
      \end{cases}
      \]

      \item \emph{Basic Strategy.}
      \[
      \scriptsize
      \text{action}(c, d) =
      \begin{cases}
          \textsc{stand} & \text{if } c \geq 17 \\
          \textsc{hit} & \text{if } c \leq 11 \\
          \textsc{stand} & \text{if } c \geq 12 \land 2 \leq d \leq 6 \\
          \textsc{hit} & \text{otherwise}
      \end{cases}
      \]
  \end{enumerate}

  \section{Evaluation Training Times}\label{app:timing}


   \begin{table}[H]
   \vspace{-0.5cm}
  \small
  \centering
  \begin{tabular}{llrrrrrrrrr}
  \toprule
  \textbf{Game} & \textbf{Method} & \textbf{4} & \textbf{8} & \textbf{12} & \textbf{16} & \textbf{20} & \textbf{50} & \textbf{100} & \textbf{500} & \textbf{1000} \\
  \midrule
  \multirow{4}{*}{\textsc{FrozenLake}}
    & TSL$_f$    & 3.1  & 2.2  & 3.0  & 3.8  & 5.0  & --   & --   & --  & --  \\
    & Alergia    & $<$.01 & $<$.01 & $<$.01 & $<$.01 & $<$.01 & .02 & .03 & .18 & .27 \\
    & BehavClone & .23 & .54 & 1.1  & 1.6  & 1.3  & 3.1  & 5.0  & 29  & 58  \\
    & CART       & $<$.01 & $<$.01 & $<$.01 & $<$.01 & $<$.01 & $<$.01 & $<$.01 & .09 & .13 \\
  \midrule
  \multirow{4}{*}{\textsc{CliffWalking}}
    & TSL$_f$    & 14   & 8.9  & 9.8  & 12   & 11   & --   & --   & --  & --  \\
    & Alergia    & $<$.01 & $<$.01 & $<$.01 & $<$.01 & $<$.01 & .05 & .02 & .09 & .24 \\
    & BehavClone & .84 & .58 & .83 & 1.2  & 1.3  & 3.6  & 8.4  & 45  & 98  \\
    & CART       & $<$.01 & $<$.01 & $<$.01 & $<$.01 & $<$.01 & $<$.01 & $<$.01 & .06 & .14 \\
  \midrule
  \multirow{4}{*}{\textsc{Taxi}}
    & TSL$_f$    & 9.4  & 8.9  & 14   & 12   & 13   & --   & --   & --  & --  \\
    & Alergia    & $<$.01 & $<$.01 & $<$.01 & $<$.01 & $<$.01 & .02 & .04 & .17 & .37 \\
    & BehavClone & 1.2  & .70 & 1.2  & 1.6  & 1.9  & 4.5  & 9.0  & 54  & 107 \\
    & CART       & $<$.01 & $<$.01 & $<$.01 & $<$.01 & $<$.01 & $<$.01 & .01 & .29 & .32 \\
  \midrule
  \multirow{4}{*}{\textsc{Blackjack}}
    & TSL$_f$    & .72 & .34 & .35 & .35 & 1.2  & --   & --   & --  & --  \\
    & Alergia    & $<$.01 & $<$.01 & $<$.01 & $<$.01 & $<$.01 & --   & --   & --  & --  \\
    & BehavClone & .74 & .15 & .16 & .16 & .35 & --   & --   & --  & --  \\
    & CART       & .36 & $<$.01 & $<$.01 & $<$.01 & $<$.01 & --   & --   & --  & --  \\
  \bottomrule
  \end{tabular}
  \vspace{0.5em}
  \caption{Training time (seconds) across all environments. Alergia and CART train in milliseconds; BehavClone scales linearly with $n$. TSL$_f$ requires 3--14s but achieves perfect generalization with 1--2 orders of magnitude
  fewer samples.}
  \label{tab:timing_all}
  \vspace{-0.5cm}
  \end{table}

    We report training times (in seconds) for all methods across environments.
  TSL$_f$ function discovery and specification mining, alongisde other baseline learning, was conducted locally on an Apple M1
  Pro processor. Controller synthesis via Issy was performed on Modal cloud
  infrastructure (x86-64, 4 vCPUs, 4GB RAM), averaging 21 minutes for FrozenLake,
  38 minutes for CliffWalking, 29 minutes for Taxi, and 3 minutes for Blackjack.